\documentclass[authoryear,preprint,review,12pt]{elsarticle}

\usepackage[a4paper,left=2cm,right=2cm,top=2.5cm,bottom=2.5cm]{geometry}
\usepackage[utf8]{inputenc}
\usepackage[T1]{fontenc}

\usepackage{amsmath,amssymb, mathtools}
\usepackage{amsthm}
\usepackage{bbm,amscd}
\usepackage{color,xcolor}
\usepackage{graphicx}
\usepackage[round,authoryear]{natbib}
\usepackage{lineno} 
\usepackage{setspace} 
\usepackage{booktabs} 
\usepackage{longtable} 
\usepackage{array} 

\definecolor{darkgreen}{rgb}{0.0, 0.2, 0.13}
\definecolor{lightbrown}{rgb}{0.71, 0.4, 0.11}

\newcommand{\E}[1]{\mathbb{E}\left[#1\right]}

\newcommand{\Var}[1]{\mathbb{V}\left[#1\right]}

\newcommand{\fall}[2]{#1^{\underline{#2}}}
\newcommand{\StirlingTwo}{\genfrac{\{}{\}}{0pt}{}} 
\newcommand{\bs}[1]{\boldsymbol{#1}}

\newcommand{\EE}{\mathbb{E}}
\newcommand{\NN}{\mathbb{N}}
\newcommand{\PP}{\mathbb{P}}
\newcommand{\RR}{\mathbb{R}}
\newcommand{\VV}{\mathbb{V}}

\newcommand{\cI}{\mathcal{I}}
\newcommand{\cO}{\mathcal{O}}
\newcommand{\scO}{{\scriptstyle\mathcal{O}}}
\newcommand{\FallingFact}[2] {{#1}^{\underline{#2}}}
\newcommand{\num}{\Omega} 
\newcommand{\defeq}{\mathrel{\mathop:}=}
\newcommand{\eqdef}{=\mathrel{\mathop:}}
\newcommand{\dd}{\,\mathrm{d}}

\theoremstyle{plain}

\newtheorem{prop}{Proposition}[section]
\theoremstyle{definition}

\journal{Theoretical Population Biology}

\begin{document}


\begin{frontmatter}

\title{Evolving genealogies in cultural evolution, the descendant process, and the number of cultural traits}

\author[1]{Joe Yuichiro Wakano \corref{cor1}}
\affiliation[1]{organization={School of Interdisciplinary Mathematical Sciences},
addressline={4-21-1 Nakano, Nakano},
postcode={164-8525},
city={Tokyo},
country={Japan}}
\ead{joe@meiji.ac.jp}

\author[2]{Hisashi Ohtsuki}
\affiliation[2]{organization={Research Center for Integrative Evolutionary Science, SOKENDAI (The Graduate University for Advanced Studies)},
addressline={Shonan Village, Hayama},
postcode={240-0193},
city={Kanagawa},
country={Japan}}

\author[3]{Yutaka Kobayashi}
\affiliation[3]{organization={Kochi University of Technology},
addressline={2-22 Eikokuji, Kochi City},
postcode={780-8515},
city={Kochi},
country={Japan}}

\author[4]{Ellen Baake}
\affiliation[4]{organization={Faculty of Technology},
addressline={Bielefeld University},
postcode={33501},
city={Bielefeld},
country={Germany}}

\cortext[cor1]{Corresponding author}


\begin{abstract}
We consider a Moran-type model of cultural evolution, which describes how traits emerge, are transmitted, and get lost in populations. Our analysis focuses on the underlying cultural genealogies; they were first described by \citet{aguilar2015} and are closely related to the ancestral selection graph of population genetics, wherefore we call them \emph{ancestral learning graphs}. We investigate their dynamical behaviour, that is, we are concerned with \emph{evolving genealogies}. In particular, we consider the total length of the genealogy of the entire population as a function of the (forward) time where we start looking back. This quantity shows a sawtooth-like dynamics with linear increase interrupted by collapses to near-zero at random times. We relate this to the metastable behaviour of the stochastic logistic model, which describes  the evolution of the number of ancestors as well as the number of descendants of a given sample. 

We superpose types to the model by assuming that new inventions appear independently in every individual, and all traits of the cultural parent are transmitted to the learner in any given learning event. The set of traits of an individual then agrees with the set of innovations along its genealogy. The properties of the genealogy thus translate into the properties of the trait set of a sample. In particular, the moments of the number of traits are obtained from the moments of the total length of the genealogy.
\end{abstract}


\end{frontmatter}




\section{Introduction}
\emph{Cultural traits} originate by innovation and are transmitted by learning, as opposed to genetic traits, which originate by mutation and are passed on via inheritance. \emph{Cultural evolution} is about how cultural traits emerge,  are  transmitted, and lost in populations. The quantitative description and mathematical modelling of cultural evolution has a venerable history, as testified, for example, by the classical monographs of \citet{CavalliSforzaFeldman1981} and \citet{BoydRicherson1985}. \emph{Stochastic} models of cultural evolution are, however, relatively recent. Here, models of population genetics are often used   as  blueprints, also with respect to the  questions asked and the methods employed.
For example, \citet{strimling2009tpb}, \citet{fogarty2015,fogarty2017}, and \citet{aoki2018} studied a discrete-time Moran-type model of cultural evolution under various assumptions on the origination and transmission of traits; they investigated quantities such as the number of traits maintained in a population and the corresponding popularity spectrum, that is, the  histogram of trait frequencies in the population,  akin to the site-frequency spectrum in population genetics. These and similar results are centered around the traits, focus on the stationary state, and characterise it by means of expectations of various quantities related to trait frequencies.

\citet{aguilar2015} were, to the best of our knowledge,   the first  to explicitly consider the cultural \emph{genealogies} in this kind of model. They worked with a continuous-time Moran-type model of cultural evolution, where the genealogy is closely related to the ancestral selection graph (ASG) of population genetics (\citet{NeuhauserKrone1997,KroneNeuhauser1997}; see also \citet{Mano2009}). \cite{aguilar2015} investigated the expected time to the most recent unique ancestor (MRUA), a quantity different from both the  most recent common ancestor (MRCA) and the ultimate ancestor (UA) of the ASG.  \citet{KobayashiWakanoOhtsuki2018} worked with a  discrete-time Wright--Fisher-type model of cultural evolution  assuming independent origination and transmission of traits and used genealogical thinking to obtain the expected number and age of distinct cultural traits in a finite sample. Altogether, little is known beyond expectations. 

The goal of this paper is twofold. We will first focus on the cultural genealogies, which are  of independent interest. In particular, we consider the total length of the genealogies (that is, the sum of the lengths of all its branches) as a function of the time where we start looking back, that is, we consider \emph{evolving genealogies}. Exploiting the connection to recent variants of the ASG of population genetics and to the stochastic logistic model of epidemiology and theoretical ecology, we analyse both the moments and the temporal fluctuations of the length of the genealogies at stationarity. Second, we superpose the traits on the model via  admittedly simplistic assumptions, namely  1) traits originate independently, 2) all traits of a cultural parent are jointly transmitted to the learner, and 3) traits are lost at death only. This way, the distribution of the number of traits is closely tied to the  genealogies,  but contains  additional randomness.

The paper is organised as follows. In Section 2, we describe and define our model forward in time and then introduce the genealogical process, which we term the \emph{ancestral learning graph (ALG)}; we also elaborate on the connections to population genetics and  epidemiology. 
In Section 3, we present the results. We first consider the moments of both the length of the genealogy and  the number of traits at stationarity, and then turn to their distributions. From there, we move on to the dynamics of the ALG, first when started from a fixed time and then as a function of the forward time where we start looking back.
In simulations, we explore the dynamics of the length of the resulting evolving genealogy (and hence the number of traits)
and observe a surprising sawtooth-like behaviour.  
We understand it by investigating the evolving genealogies (backward in time) and the corresponding descendant process (forward in time) with the help of the properties of the stochastic logistic process. In particular, the metastable state of the latter and the merging of ALGs will play crucial roles. As an aside, we also obtain the popularity spectrum. We discuss our findings in Section 4.


\section{The model}

We follow the  model of \citet{aguilar2015}, who assume a population of  fixed size $N$ in continuous time with death-birth  and learning events. We primarily think of individuals as humans, but they could be reinterpreted as groups of humans, see the discussion. The model bears elements of both the  Moran model of population genetics and the SIS (susceptible-infected-susceptible) model of epidemiology; we will discuss these connections later.  We assume that newborn individuals do not have any cultural parent at the time of birth, but acquire role models one by one through learning events, each of which assigns a (single) new cultural parent to the focal individual. \citet{aguilar2015} restrict themselves to this so-called untyped model, where it is  intentionally left unspecified which traits are transmitted in a given learning event, and whether and how new traits emerge. We combine this  with a simplistic model that makes the appearance and transmission of traits explicit. This is motivated by the work of \citet{KobayashiWakanoOhtsuki2018},   who assume a discrete-time Wright--Fisher process, where every newborn invents a random number of new traits and is assigned a random number $K$ of potential role models; every trait of each of the role models is independently transmitted to the newborn with  probability $b$. Here we assume that, within the continuous-time model of \citet{aguilar2015},  new traits appear at constant rate $\mu$ in every individual, and we specify the mode of  transmission in accordance with \citet{KobayashiWakanoOhtsuki2018} with the choice $b=1$, so  the learning individual acquires \emph{all} traits of its cultural parent; this assumption is made for mathematical convenience.
The resulting model may also be considered as a continuous-time version of (special cases of) the models of \citet{strimling2009tpb}, \citet{fogarty2015, fogarty2017}, and \cite{aoki2018}. 
In the discussion, we will --- as an outlook --- waive the admittedly unrealistic assumption of all-or-none transmission  and  explore what happens if every trait of the role model is independently transmitted to the learner with probability $b<1$.


\subsection{The forward process}
Consider a population of  $N$ individuals. We will work with   the \emph{graphical representation} as introduced by \citet{Harris1978} (see \cite{Liggett2011} for a review). This is now an indispensable tool in stochastic models in biology and physics and, for our case, shown in Figures~\ref{learning-1-trait-ips} and \ref{learning-ips-innov}. Each individual corresponds to a horizontal line segment, with the forward direction of time being left to right; the lines have labels in $[N]\defeq\{1,2,\ldots,N\}$ and extend through all of $\RR$. The events described below are represented by graphical elements superposed to this picture. Again in line with a common strategy for the Moran model, we describe the model in two steps: first the  untyped model, which contains the death-birth  and learning events; and second the typed model, which includes the traits, specifically, the innovation events and the transmission of traits on the occasion of learning events. 
\begin{enumerate}
\item \emph{Untyped version.}
Every individual can experience two kinds of events at any instant of time and independently of the remaining population: a death event and a learning event. Death  occurs at rate $u>0$ per individual, and when this happens, a newborn will replace it. 
Learning events occur at rate $s>0$ per individual. In a learning event, the individual randomly samples a cultural parent uniformly from among the $N$ individuals (including itself). So far, the learning events are abstract; we do not yet say what is transmitted along them. In the graphical representation (Fig.~\ref{learning-1-trait-ips}), death-birth events are symbolised by crosses; they appear
at rate~$u$ on every line, by way of independent Poisson point processes. Learning events are depicted by arrows between the lines, with the cultural parent at the tail and the learner at the tip; arrows appear at rate $s/N$ per ordered pair of lines, again by way of independent Poisson point processes on $\RR$. We omit arrows where the parent and the learner are identical.

\begin{figure}
\begin{center}
\includegraphics[width=70mm]{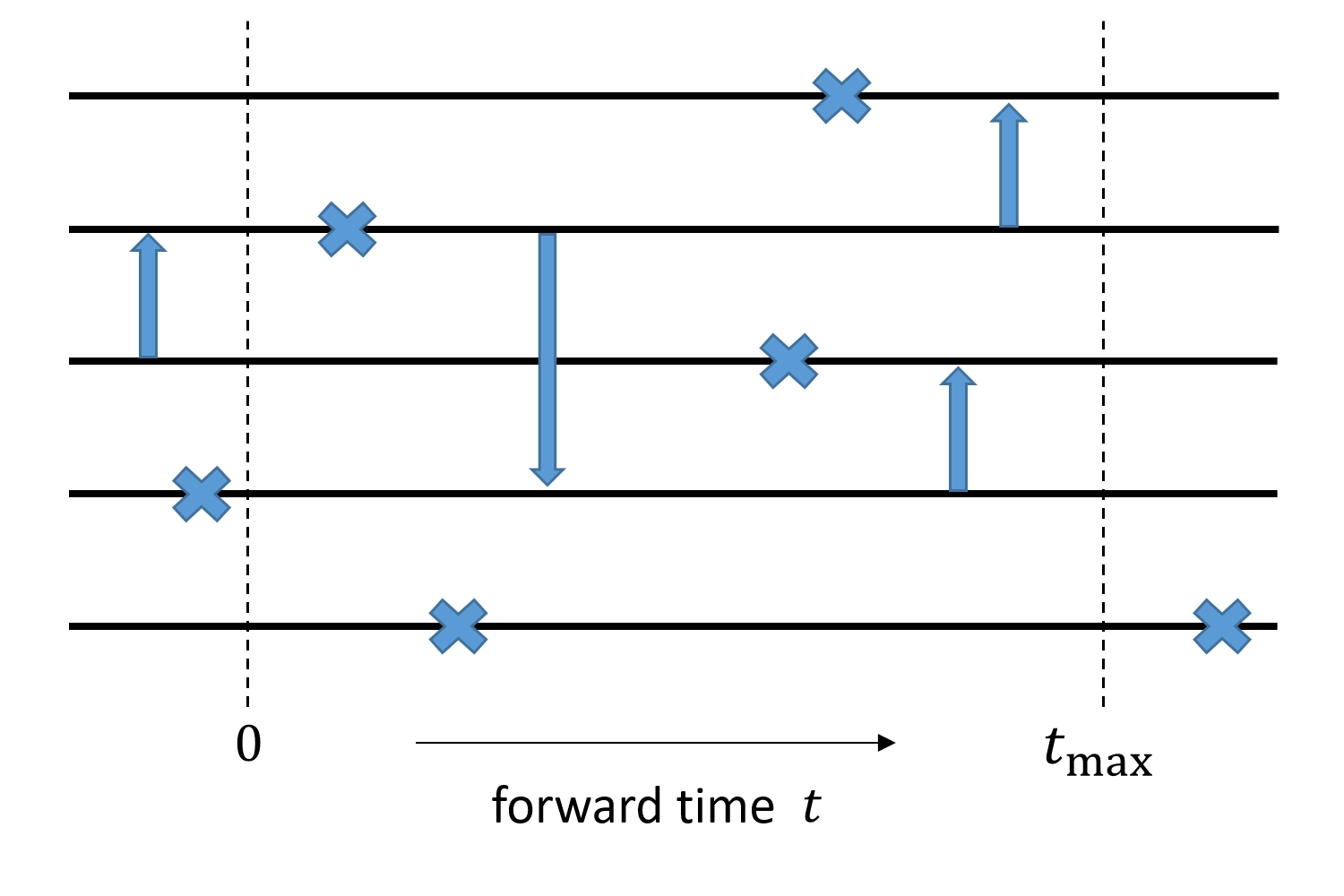}
\label{learning-1-trait-ips} \end{center}
\caption{Graphical representation of a realisation of the untyped  model ($N=5$). The lines extend through all of $\RR$ and are decorated with crosses (death-birth events) and learning arrows. In this way, the time evolution is defined for any time interval such as $[0,t_{\text{max}}]$.
}
\end{figure}

\item \emph{Including the types.} The type of an individual is the collection of its cultural traits. Given a realisation of the untyped system, we  turn it into a typed one as folllows. We assume that every individual independently  invents  new traits at rate $\mu>0$. In the graphical representation, these innovations are indicated by circles, which are laid down on every line by way of a Poisson point process on $\RR$ at rate $\mu$, see Figure~\ref{learning-ips-innov}. We assume that, first, all innovations are non-recurrent, that is, have never occurred in the population before. Second, in  every learning event, the learner acquires all traits carried by the cultural parent; this happens by merging the knowledge sets of the learner and the parent. Third, we assume that death events eliminate all traits, so the newborn is devoid of any trait. 

We now assign a set of traits to every line at time $t=0$ in an exchangeable way (that is, according to a law that is invariant under permutation of lines); this set may or may not be  empty. From $t=0$ onwards, we then number the innovations consecutively in the order of their appearance. If we finally propagate the traits forward in time according to the rules just described, we get the types, that is, the collection of traits, for every individual and any time $t>0$.
Let now $K_{\alpha}(t) \subset \mathbb{N}$ be the set of traits that individual $\alpha \in [N]$ knows at time $t \in \mathbb{R}_{\geq 0}$; that is, the individual's \emph{knowledge set}. For a subset of individuals $G \subseteq [N]$, we define 
\begin{equation}
    K_{G}(t) \defeq \bigcup_{\alpha \in G} K_{\alpha}(t)
\end{equation}
by slight abuse of notation (so $K_{\{\alpha\}}(t)=K_{\alpha}(t)$).
With this, we denote the number of traits of a `typical'  group of $n$ individuals at time $t$, that is, the cardinality of the group's knowledge set, as
\begin{equation}\label{def_C_t(n)}
    C_{n}(t) \defeq |K_{[n]}(t)|.
\end{equation}
By `typical', we here allude to the fact that,
by exchangeability, the $K_G(t)$ are identically distributed for all $G \subseteq [N]$ with $\lvert G \rvert = n$,  so we choose $G=[n]$ as their representative. Clearly, then, $C_{n+1}(t) \ge C_n(t)$ for all $t$.
We will see below that the distribution of $(C_{n}(t))_{t \geq 0}$ will become stationary as $t \to \infty$; we will denote by  $C_{n} = C_{n}(\infty)$ a random variable that has this stationary distribution.

\end{enumerate}

\begin{figure}
\begin{center}
\includegraphics[width=70mm]{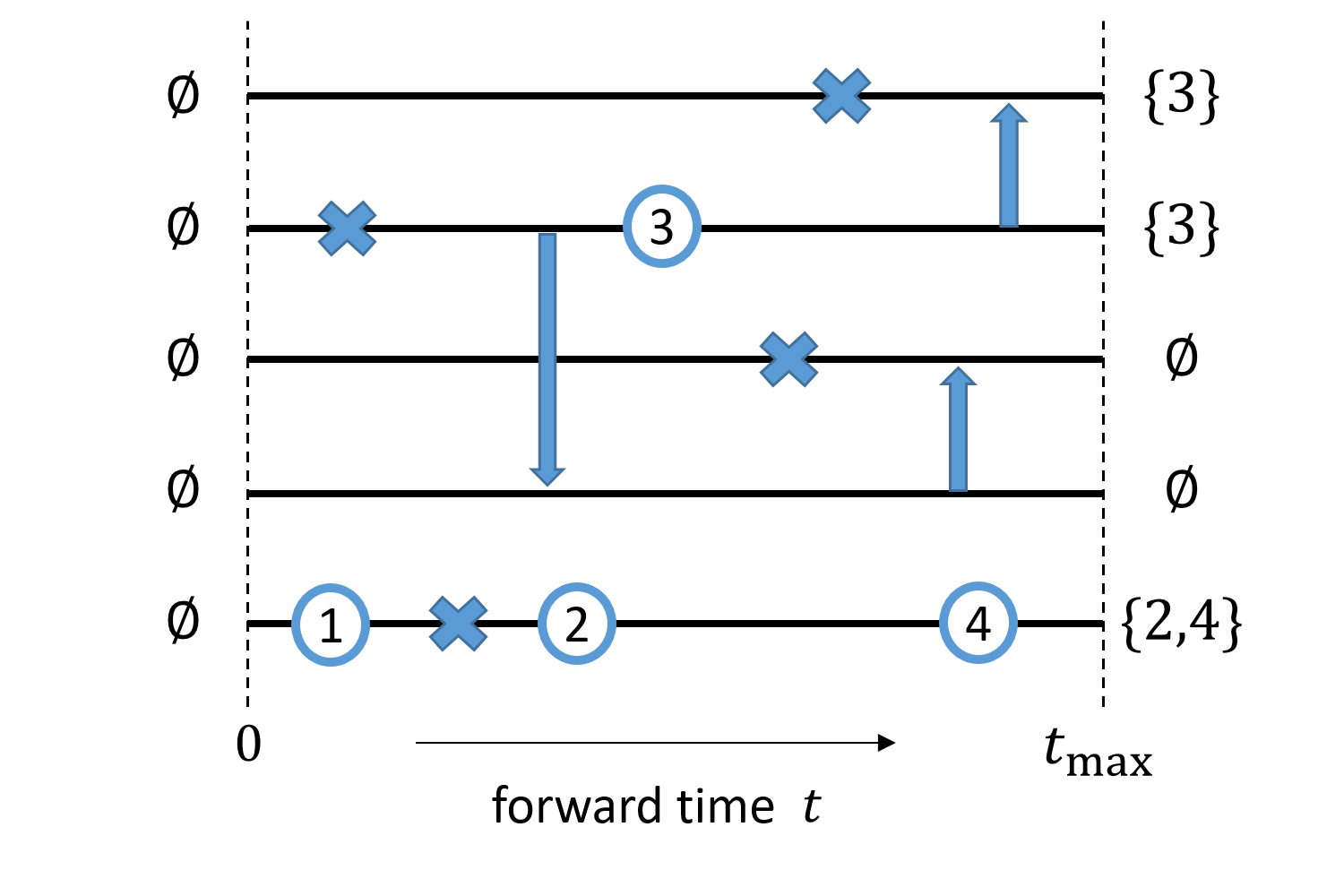}
\end{center}
\caption{\label{learning-ips-innov} A typed version of Figure~\ref{learning-1-trait-ips} in the interval $[0, t_{\text{max}}]$. 
}
\end{figure}


\subsection{Formal definition of the forward process}\label{subsec:FullForward}
Let
\begin{equation}
\Phi(t) \defeq \big ( K_1(t), \ldots,K_N(t),\num(t) \big ) 
\end{equation}
be the state  of the typed graphical construction at time $t$, where $\num(t)\in \mathbb{N}_0$ denotes the number of traits ever invented until $t$, including those assigned at time 0; so $\Omega(0)=\lvert \cup_{\alpha \in [N]} K_\alpha (0) \rvert$ and $K_\alpha(t)\subseteq \{1,2,\dots,\num(t)\}$ (where the latter is agreed to be $\varnothing$ if $\Omega(t)=0$), and
\begin{equation}
    \underbrace{2^{\mathbb{N}} \times \dots \times 2^{\mathbb{N}}}_{N \text{ times}} \times \mathbb{N}_{0}   
\end{equation}
is the state space. Note that 
$\num(t)=\lvert \cup_{t^\prime \in [0,t]}\cup_{\alpha \in [N]} K_\alpha (t^\prime) \rvert$. 
The process starts at
$\Phi(0)=(K_1(0), \ldots, K_N(0),\num(0))$ and evolves as follows. If $\Phi(t)=(k_1,\ldots,k_N,\omega)$,  the following events may happen: 

At rate $Nu$, a death event occurs and a uniformly-chosen individual $\alpha \in [N]$ loses all traits, so  $k_\alpha$ changes to $\varnothing$, and we see the transition
\begin{subequations}
\begin{equation}
    (k_1,\ldots,k_{\alpha}, \ldots, k_N,\omega) \longrightarrow (k_1,\ldots, \varnothing, \ldots,k_N,\omega)
\end{equation}
with $\varnothing$ in position $\alpha$.
At rate $N\mu$, an innovation event occurs and a uniformly-chosen individual $\alpha$ acquires a new trait, so  $\omega$ changes to $\omega+1$, $k_\alpha$ changes to $k_\alpha \cup \{\omega+1\}$, and the transition is
\begin{equation}
    (k_1,\ldots,k_{\alpha}, \ldots, k_N,\omega) \longrightarrow (k_1,\ldots, k_{\alpha} \cup \{\omega+1\}, \ldots,k_N,\omega+1).
\end{equation}
At rate $Ns$, a learning event occurs and a uniformly-chosen individual $\alpha$ learns from another uniformly-chosen individual $\beta$, so  $k_\alpha$ is replaced by  $k_\alpha \cup k_\beta$ and 
\begin{equation}
    (k_1,\ldots,k_{\alpha}, \ldots, k_N,\omega) \longrightarrow (k_1,\ldots, k_{\alpha} \cup k_\beta , \ldots,k_N, \omega);
\end{equation}
\end{subequations}
note that nothing happens if $\alpha=\beta$.
Note also that we have defined the process in a way that is close to how we will later simulate it; in this formulation, it does not become stationary for $t \to \infty$ since, due to the consecutive labelling of the traits, both $(\num(t))_{t \geq 0}$ and the $(K_{\alpha}(t))_{t \geq 0}$  are transient. The quantities we will analyse, however, will only rely on the number of traits rather than their labels and converge to stationarity as $t \rightarrow \infty$, as will become obvious in the backward picture.


\subsection{The ancestral learning graph (ALG)}
Let us now take a backward perspective and consider the genealogy of  $n \in [N]$ individuals sampled at forward time $t$, to which we will refer as the present; we may, but need not, choose $t$ as some fixed final time $t_{\text{max}}$. Starting again from the untyped version of the graphical representation, let us trace back the ancestral lines of the sample,  as illustrated in Figure~\ref{alg-1-trait}. More precisely,  we first describe the untyped version of the \emph{ancestral learning graph (ALG)}, which consists, at any given time, of  the set of all cultural parents, parents of parents and so forth, that is, of all cultural ancestors,  of the sample.  It is constructed as follows. 

\begin{figure}
\begin{center}
\includegraphics[width=70mm]{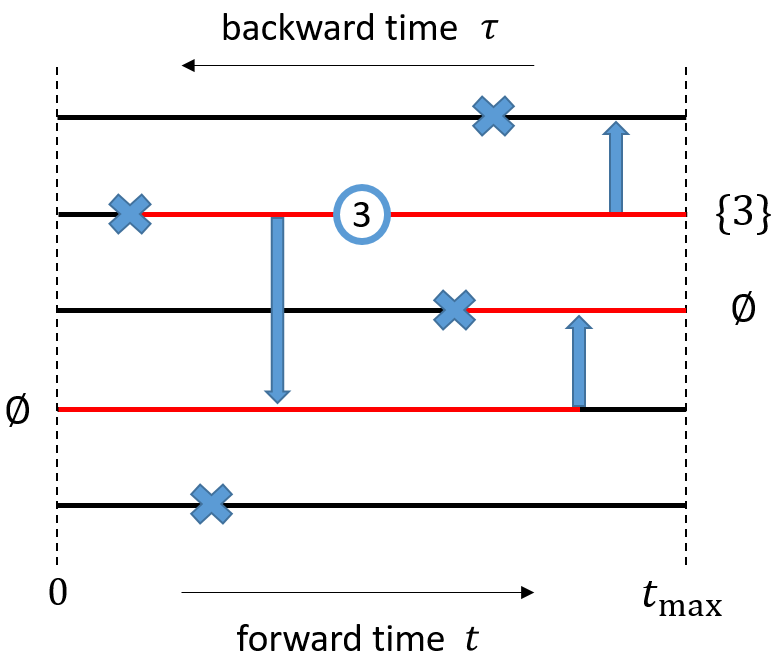}
\end{center}
\caption{\label{alg-1-trait} An ALG embedded in the realisation of the forward model in Figure \ref{learning-ips-innov}, starting from $t=t_{\text{max}}$. Red and black lines are ancestral and non-ancestral, respectively, to the sample of size $n=2$ taken at $\tau=0$.}
\end{figure}

\begin{enumerate}
\item
\label{traceback}
Start the graph by tracing back the lines emerging from the  individuals sampled at forward time $t$.  Denote backward time by $\tau$, so backward time $\tau$ corresponds to forward time $t-\tau$; in particular,  forward times $0$ and $t$ correspond to backward times $\tau=t$ and $\tau=0$, respectively. Proceed as follows in an iterative way in the backward direction of time until the initial time is reached, or until all ancestral lineages have been hit by a cross. 
\begin{enumerate}
\item 
\label{branching}
 If a line currently  in the graph is hit by the tip of a learning arrow, we trace back both  ancestors, namely that of the cultural parent  (at the tail of the arrow; in line with  ASG terminology, we call it the incoming branch) and that of the learning individual  (at the tip; again, as in the ASG, we call it the continuing branch). That is, we add the incoming line to the graph; this results in a \emph{branching event}. Note that the  graph remains unchanged if it already contains the incoming branch. 
\item
 If a line currently  in the graph is hit by a cross, we do not follow it  any further, that is, we prune it.
\end{enumerate}
The resulting \emph{untyped ALG} consists of all cultural ancestors of the sampled individuals over time, including the individuals in the sample themselves until their births (that is, at branching events, individuals are counted as their own cultural ancestors); the untyped ALG  corresponds to the untyped version of the forward model. 
\item The untyped ALG may  be turned into a \emph{typed} one by first assigning, to each line in the graph that is still alive at backward time $t$, the initial trait set from the forward model; if there are no lines left at backward time $t$, this step is void. Then, the traits are propagated  forward (that is, to the right) in the genealogy according to the same rules as before. That is, when a new trait was acquired in the forward model on a line belonging to the genealogy, it is added to the trait set of the individual; in each branching event, the trait set of the parent is united with that of the learner; and if a line encounters a cross,  it loses all traits. This way, a trait set is associated with every line element of the graph.
\end{enumerate}

Note that the untyped ALG obtained in Step 1 is the cultural genealogy of \citet{aguilar2015}, with the only difference that we allow for individuals to also choose themselves as  cultural parents; but this is an event without any effect, so it only changes the time scale by a factor of $(N-1)/N$.
The typing in Step 2 describes how traits appear and are transmitted along the given genealogy; this step was intentionally left unspecified by \citet{aguilar2015}. Note also that we can obtain a \emph{stationary} situation for any $t \geq 0$ by running the untyped ALG until all (ancestral) lines have been hit by a cross, even if this requires $\tau > t$; this is where the graphical construction on all of $\RR$ comes in handy. Turning the untyped ALG into a typed one then does not require the assignment of initial types --- the types are solely determined by origination and transmission of traits along the graph.


\subsection{The ALG as a stochastic process\label{AppendixBackwardSimulation}}
So far, we have described the ALG as a function of a given realisation of the graphical representation. We now define it as a stochastic process  independent of such a realisation, again first in the untyped  
and then in the typed version. 
\begin{enumerate}
\item
For the untyped version, let $\Lambda^t_\alpha(\tau) \subseteq [N]$ denote the set of ancestors at backward time $\tau \geq 0$ of individual $\alpha \in [N]$ at forward time $t \geq 0$; note that we allow $\tau > t$ as before.
For $G \subseteq [N]$, set $\Lambda^t_G(\tau) \defeq \cup_{\alpha \in G} \Lambda^t_\alpha(\tau)$ by slight abuse of notation (so $\Lambda^t_{\{\alpha\}}(\tau) = \Lambda^t_\alpha(\tau)$). Then the process $(\Lambda(\tau))^t_{\tau \geq 0} \defeq \big (\Lambda^t_1(\tau),\ldots,\Lambda^t_N(\tau)\big)_{\tau \geq 0}$ contains the complete genealogical information of all individuals in the present population. It evolves as follows. Clearly $(\Lambda^t_1(0), \ldots, \Lambda^t_N(0)) = (\{1\},\{2\}, \ldots, \{N\})$. If, for any time $\tau \geq 0$, the current state is $\Lambda^t(\tau) = (\lambda_1, \ldots, \lambda_N)$, the following events may happen:  at rate $Nu$, a death-birth event hits a randomly-chosen individual $\beta$. So $\beta$ is removed from the set of ancestors of anybody, and
\begin{subequations}
\begin{equation}
    (\lambda_1,\ldots, \lambda_N)
    \longrightarrow (\lambda_1 \setminus \{\beta\},\ldots, \lambda_N \setminus \{\beta\}).
\end{equation}
At rate $Ns$, a learning event occurs and a randomly-chosen individual $\beta_1$ learns from another randomly-chosen individual $\beta_2$ (where $\beta_2 = \beta_1$ is allowed, in which case the event is silent). So if, for $\alpha \in [N]$, $\beta_1$ is contained in $\lambda_\alpha$, then $\beta_2$ joins this ancestral set, and otherwise nothing happens:
\begin{equation}
    \lambda_{\alpha}
    \longrightarrow \begin{cases} \lambda_\alpha \cup \{\beta_2\}, & \text{if } \beta_1 \in \lambda_\alpha, \\ \lambda_{\alpha}, & \text{otherwise}, \end{cases} \qquad \alpha \in [N].
\end{equation} \label{eq:lambdaTransitions}
\end{subequations}
In words, any individual $\alpha$ who has $\beta_1$ as its ancestor just before the focal learning event will add $\beta_2$ as another ancestor as a result of the event.

\item 
The untyped ALG may  be turned into a \emph{typed} one by superposing the untyped graph with a Poisson point process that lays down innovations at rate $\mu$ on every line segment. Together with an initial assignment of types to the lines alive at backward time $t$ (recall this step is void if we run the graph until there are no lines left) and propagation along the graph as in the forward model, this determines the types along all lines in the graph.
\end{enumerate}

Returning to the untyped version, let us note for later use that, for any $n \in [N]$ and $\alpha \in [n]$, we have  $\Lambda^t_\alpha(\tau) \subseteq \Lambda^t_{[n]}(\tau)$ for all $t,\tau \geq 0$ and so
\begin{equation}\label{identical_sets}
    |\Lambda^t_\alpha(\tau_1)|=|\Lambda^t_{[n]}(\tau_1)|
\quad \text{for some $\tau_1\geq 0$ implies}\quad
 \Lambda^t_\alpha(\tau) = \Lambda^t_{[n]}(\tau) \hspace{10pt}
    \text{for all } \, \tau \ge \tau_1.
\end{equation}

In what follows, an important role will be played by the \emph{line-counting process} $(Y^t_n(\tau))_{\tau \geq 0}$ of the ALG. That is,
\begin{equation}\label{Y^t_n_Lambda_n} 
Y^t_n(\tau) \defeq \lvert \Lambda^t_{[n]}(\tau)\rvert
\end{equation}
is the number of lines at backward time $\tau$ of the sample $[n]$ of individuals taken at forward time $t$; due to exchangeability, we  allow ourselves to simply speak of a sample of size $n$. Also, we will sometimes omit the dependence on   $n$. 
For any given $t$, $(Y^t(\tau))_{\tau \geq 0}$ is  a birth-death process\footnote{This \emph{birth-death process} is not to be confused with the \emph{death-birth events} in the model, as indicated by crosses in the graphical representation.} on 
$[N]_0:= \{0,1,\ldots,N\}$ with initial value $Y^t_n(0)=n$ and birth and death rates\footnote{not to be confused with the realisations $\lambda_\alpha$ of $\Lambda_\alpha$ and the innovation rate $\mu$.} 
\begin{equation}\label{bdrates}
\lambda_y \defeq sy \frac{N-y}{N} \quad \text{and} \quad \mu_y \defeq uy,
\end{equation}
respectively, when $Y^t(\tau)=y$ (note that the factor $(N-y)/N$ is the probability that a learning arrow comes from outside the current graph). 
Clearly, $(Y^t(\tau))_{\tau \geq 0}$  is an absorbing Markov chain with 0 as the only absorbing state. The behaviour of the process  is very well studied, since it is, at the same time, the \emph{stochastic logistic process} of ecology and epidemiology; in particular, $(Y^t_n(\tau))_{\tau \geq 0}$ is the number of infected individuals in a stochastic  SIS model  with infection rate $s$ and recovery rate $u$ per individual when started with $n$ infected individuals at time $t$ \cite[Ch.~8.2]{AnderssonBritton2000}. 

As alluded to already, another crucial quantity will be $L^t_{n}$, the \emph{(total) length of the genealogy}, that is, the total length of all branches, as $\tau \to \infty$, in the ALG of a sample of size $n$ taken at time $t$, that is,
\begin{equation}\label{Ln_in_path-integral}
    L^t_{n} \defeq \lim_{\tau \to \infty} L^t_{n}(\tau), \quad \text{where}  \quad L^t_{n}(\tau) \defeq \int_{0}^{\tau} Y^t_n(r) \mathrm{d}{r} = \int_0^\tau \lvert \Lambda^t_{[n]}(r) \rvert \mathrm{d} r = \int_{t-\tau}^t \lvert \Lambda^t_{[n]}(t-s) \rvert \mathrm{d} s.
\end{equation}
Indeed, the limit exists and is  finite for almost all realisations of the ALG, since $(Y^t_n(\tau))_{\tau \geq 0}$ absorbs in 0 almost surely in finite time for any given $n$. In  epidemiology, the quantity $L^t_{n}$, sometimes up to a constant, is  referred to as the \emph{cost of an epidemic}, because it measures the total time that individuals spend in the infected state in one episode of the epidemic started with $n$ individuals at time $t$ (\citealp{jerwood1970note, mcneil1970integral, downton1972area, gani1972cost, ball1986unified, SuarezChavez1999}; \citealp[Secs.~2.2 and 2.4]{AnderssonBritton2000}; \citealp{crawford2018computational}). Previous work has concentrated on its mean $\mathbb{E}[L^t_{n}]$ or its Laplace transform $\mathbb{E}[e^{-\theta L^t_{n}}]$ (or its moment-generating function $\mathbb{E}[e^{\theta L^t_{n}}]$) but, to the best of our knowledge, explicit expressions for the variance (or higher-order moments) are unavailable so far (but see \citet{StefanovWang2000} for a numerical evaluation of the variance).

Let us emphasise that we have defined the ancestral processes \emph{as functions of the forward time $t$ where we start looking back}; this will become important in Section~\ref{sec:EvolvingGenealogies}. For the moment, we fix $t=t_{\text{max}}$ and, when we do so, omit the superscript, so $\Lambda_\alpha(\tau) \defeq \Lambda^{t_{\text{max}}}_\alpha(\tau)$ and likewise for the other quantities.


\subsection{The connection with models of population genetics}

We have already mentioned the connection of the ALG with the SIS model and, in particular, the stochastic logistic process; this will be used extensively later. The genealogical point of view, however, is less important in epidemiology; notably, in epidemiology, the stochastic logistic process usually appears in forward-time models, whereas our $(Y(\tau))_{\tau \geq 0}$ is an ancestral process. It is therefore worthwhile to investigate the somewhat subtle connection with population genetics theory more closely, where the  backward perspective is crucial and highly developed. 

We claim that our learning model is closely related to, but not identical with, a Moran model with selective reproduction and deleterious mutation, but \emph{without} resampling (see \citet{BaakeWakolbinger2018} for a review of the Moran model and its ancestral structures under the various evolutionary forces mentioned). To see this, let us first consider the Moran model with selective reproduction, deleterious mutation, and \emph{with} resampling, and then compare it to the learning model. In this model, selective arrows appear between every ordered pair of lines at rate $s/N$; both the incoming and the continuing branches are then potential genetic parents of the affected individual, with only one of them being the true parent. These arrows correspond to the learning arrows of our model. The difference is that both the incoming and the continuing lines are (true) cultural parents (recall that learners are counted as their own cultural parents); there are no potential parents. In line with this, a learner's knowledge set is not replaced by, but merges with that of the cultural parent.
Furthermore, it is important to note that, while the selective arrows in the Moran model represent (genetic) selection, the corresponding learning arrows in the ALG are (culturally) neutral since individuals are chosen as cultural parents independently of the traits they carry. Next, resampling is a central feature of the Moran model and means that selectively neutral reproduction arrows  appear at rate $1/N$ between every ordered pair of lines. They indicate that the individual at the tail is the genetic parent of the newborn individual at the tip of the arrow. Since resampling is absent in the learning model, we  speak of the latter as a \emph{Moran-type} model. Finally, deleterious mutation events in the Moran model happen at rate $u$ per line and mean that the type of that individual is `reset' to the (selectively) weakest type. The corresponding events in the learning model are the death-birth events, where an individual is replaced by an ignorant newborn.

The connection is mirrored in the corresponding ancestral structures. Indeed, the untyped ALG is identical with a finite-population variant of the untyped ASG that lacks coalescence events (because there is no resampling), and where a line is pruned if it is hit by a deleterious mutation. The  difference lies in the interpretation: the  lines of the ALG correspond to the cultural parents of the sample, whereas the lines of the pruned ASG represent the potential genetic parents. 

As a consequence, the ALG belongs into  the context of the rich recent literature on variants of the ASG and their applications. For example, \cite{PfaffelhuberPokalyuk2013} and \cite{Grevenetal2016} used the ASG to study fixation of a beneficial mutant; \cite{Cordero2017} studied the ASG in a law-of-large-numbers regime (note that here, as well, the coalescence events vanish); \cite{CorderoMoehle19}  studied the stationary number of lines under generalised coalescence mechanisms; \cite{Boenkostetal2021}  extended the ASG to the Cannings model of reproduction; \cite{CorderoVechambre23}  investigated it in a random environment; and \cite{BaakeEsercitoHummel2023}  included multiple branching, pruning individual lines, and killing the entire process in order to account for a certain kind of frequency-dependent selection, deleterious, as well as beneficial mutations.

In the typed learning model, the innovation mechanism is  similar to the mutation process in the \emph{infinite-sites model} of population genetics (see, for example, \citet[Ch.~1.2]{Wakeley2009}), where sequences of infinite length are considered and every mutation hits a site that has never mutated before.

If we reduce the number of possible traits to one and do not allow for innovations (that is, consider the limiting case $\mu=0$), we only have two types of individuals: those that have and those that do not have the trait. Let $X(t) \in [N]_0$ be the number of individuals with the trait (so $N-X(t)$ do not have it). 
If $X(t)=x$, we  have a transition to  $x-1$ at rate $ux$ (death) and  to $x+1$ at rate  $sx(N-x)/N$ (learning). This is equivalent to a two-type Moran model  with selective reproduction at rate $s$ and deleterious mutation at rate $u$, but without resampling, without frequency-dependent  selection, and without beneficial mutations; it is a special case of the model tackled by \cite{BaakeEsercitoHummel2023}. Furthermore, note that $(X(t))_{t \geq 0}$ has the same birth and death rates as $(Y(\tau))_{\tau \geq 0}$ and hence the same law.

The type-frequency process of the Moran model with resampling, frequency-dependent selection, as well as deleterious and beneficial mutation on the one hand  and the line-counting process of the corresponding ASG variant on the other hand are in factorial moment duality with each other \citep[Theorem~2.3]{BaakeEsercitoHummel2023}. This translates into the factorial moment duality between the type-frequency process of the above single-trait learning model  and the line-counting process of the corresponding untyped ALG as follows. For $z,m \in \NN_0$, let  
\begin{equation}
    \fall{z}{m} \defeq 
    \begin{cases}
        1, & m = 0, \\
        z (z-1) \cdots (z-m+1), & m \geq 1,
    \end{cases}
\end{equation}
denote the falling factorial.
The processes $(X(t))_{t \geq 0}$ and $(Y(\tau))_{\tau \geq 0}$ are dual with respect to the duality function
\begin{equation}\label{duality_function}
H(x,y) \defeq 
\frac{\FallingFact{(N-x)}{y}}{\FallingFact{N}{y}},  \quad x,y \in [N]_0,
\end{equation} 
that is, $(X(t))_{t \geq 0}$ and $(Y(\tau))_{\tau \geq 0}$ satisfy the relation
\begin{equation}\label{duality}
\mathbb{E}[H(X(t),y) \mid X(0)=x]=\mathbb{E}[H(x,Y(t)) \mid Y(0)=y]
\end{equation} 
for $x,y \in [N]_0$ and $t \ge 0$.

Note that $H(x,y)$ is the probability to obtain  only individuals without the trait when sampling $y$ individuals without replacement from a population that contains $x$ individuals with the trait (and $N-x$ individuals without it); and the duality  means that the  sampling probability at time $t$ can be obtained either via the forward or via the backward process. As a consequence, one may  obtain, and gain insight into,  properties of the learning model  forward in time  by studying its dual process. Since $(X(t))_{t \geq 0}$ and $(Y(\tau))_{\tau \geq 0}$ have the same law, they are actually self dual with respect to $H$. 

Let us mention that factorial moment dualities  have a long history in the context of fixed-size population genetic models. In the  case without selection, they already appear in papers by  \citet{Cannings74} and \citet{Gladstien77a,Gladstien77b,Gladstien78}, at a time where neither genealogical processes  nor the concept of dualities for Markov chains had been formulated yet. Rather, the dualities appear in terms of algebraic identities between matrices, and are used  to calculate the eigenvalues of the Markov transition matrices via a similarity transform; the connection with the backward point of view is at most implicit. In 1999, \citeauthor{mohle1999concept} established factorial moment dualities in neutral genetic models, explicitly and in terms of backward processes; this paper also provides general background on duality in models of population genetics.


\subsection{Law of large numbers}
\label{subsec:LLN}
Let us briefly comment on the deterministic limit, where $N \to \infty$ without rescaling of parameters or time. Then $(Y(\tau))_{\tau \geq 0}$ turns into a linear birth-death process $(\widetilde Y(\tau))_{\tau \geq 0}$ on $\NN_0$ with birth rate $s$ and death rate $u$ per individual. In contrast to the finite-$N$ case, $(\widetilde Y(\tau))_{\tau \geq 0}$ does not necessarily absorb in 0. For  $s \leq u$, 
$(\widetilde  Y(\tau))_{\tau \geq 0}$ is (sub)critical 
with $\PP(\lim_{\tau \to \infty} \widetilde  Y(\tau) = 0 \mid \widetilde  Y(0)=1)=1$;  for $s > u$, $(\widetilde  Y(\tau))_{\tau \geq 0}$ is supercritical, and $\PP(\lim_{\tau \to \infty} \widetilde Y(\tau) = 0 \mid \widetilde  Y(0)=1 )=u/s<1$. If the process does not die out, it grows to infinite size almost surely. (These are classical results from the theory of branching processes \citep[Ch.~III.4]{AthreyaNey1972}.)

In the deterministic limit of the single-trait model $(X(t))_{t \geq 0}$ as described in the previous section, the sequence of processes $(X^{(N)}(t)/N)_{t \geq 0}$ (where the upper index indicates the population size) converges, as $N \to \infty$, to the solution of the initial value problem  
\begin{equation}\label{ode_xi}
\dot \xi(t) =  \xi(t) [s(1-\xi(t))-u], \quad  \xi(0) = \xi_0 \in [0,1],
\end{equation}
provided $X^{(N)}(0)/N \to \xi_0$. The differential equation has two equilibria: one at $1-u/s$, the other at $0$. They perform a transcritical
(or exchange of stability) bifurcation at $s=u$: for $s<u$,  the equilibrium at $0$ is attracting, while the one at $1-u/s$ (which is then $<0$) is  repelling; and vice versa for $s>u$, where $1-u/s>0$. For $s=u$, the two equilibria coincide at 0, which is then attracting.

The duality \eqref{duality}  carries over to the deterministic limit; specifically, evaluating \eqref{duality} for $y=1$, $N \to \infty$ such that $x/N \to \xi_0$, and $t \to \infty$ gives for the unique stable equilibrium $\bar \xi$ of the differential equation that $1-\bar \xi = \EE[(1-\xi_0)^{\widetilde Y(\infty)}\mid \widetilde Y(0)=1]$ for $\xi_0 \in (0,1]$ and so $\bar \xi = 1-\PP(\lim_{\tau \to \infty} \widetilde Y(\tau) = 0 \mid \widetilde  Y(0)=1)$. For $s \leq u$, therefore, the trait is always absent at equilibrium, while, for $s>u$, the trait is present in a positive proportion of  individuals at the stable equilibrium. See \citet{Cordero2017}, \citet{BaakeCorderoHummel2018}, or \citet{BaakeWakolbinger2018} for the details about the deterministic limit. 

Let us note for later use that the differential equation \eqref{ode_xi} is not only the deterministic limit of our single-trait model, but also the deterministic limit  of  $(Y^{(N)}(\tau)/N)_{\tau \geq 0}$ (that is, for $N \to \infty$ in the sequence of processes $(Y^{(N)}(\tau)/N)_{\tau \geq 0}$ with population size $N$, without rescaling of parameters or time), provided that $Y^{(N)}(0)/N$ converges;
this is clear because $(X^{(N)}(t))_{t \geq 0}$ and $(Y^{(N)}(\tau))_{\tau \geq 0}$ have the same law.

In what follows,  we continue to adhere to finite $N$, but the two regimes $s<u$ and $s>u$  still behave in qualitatively different ways (as long as $s$ is not too close to $u$ and with a smooth transition between the two regimes), see \citet[Chap.~2.5]{Nasell2011} or \citet{Foxall2021}. We will therefore continue to  use the notions \emph{subcritical} 
and \emph{supercritical}, in line with the literature.


\section{Results}
The typed ALG is the appropriate genealogical structure to study the dynamics and the stationary distribution  of $(C_{n}(t))_{t \geq 0}$.  Since the typed ALG results from the untyped one by superposition of  a Poisson point process that lays down innovations at rate $\mu$ on every line, and every innovation on the genealogy is passed on to the sample, the crucial quantity is the  length  of the genealogy, $L_n$ of  \eqref{Ln_in_path-integral}. 
By the genealogical picture, 
\begin{equation}\label{Cn to Ln}
    C_{n} \sim \mathrm{Poi}(\mu L_n),
\end{equation}
that is, $C_{n}$ is a mixed Poisson random variable: a Poisson random variable with a random parameter $\mu L_{n}$.
So the genealogies, and $L_n$ in particular, are the fundamental quantities, which are also of independent interest, and $C_n$ contains additional randomness due to the Poissonisation.

We will, in what follows, consider the moments and the  distributions of $L_n$ and $C_n$ and then turn to the dynamics of $(L^t_n)_{t \geq 0}$ and $(C_n(t))_{t \geq 0}$,
throughout with an eye on the additional fluctuations induced by the innovation process.

\subsection{The moments}

\subsubsection{Path integrals of birth-death processes}
We start by observing that $L_n$ in \eqref{Ln_in_path-integral} is a special case of the path integral of a function $f: \NN_0 \to \RR_{\geq 0}$ over a general continuous-time birth-death process $(Z(t))_{t \geqslant 0}$ on $\NN_0$ with  unique absorbing state  0, namely
\begin{equation}\label{def_path-integral}
    \mathcal{I}_{i}(f) := \int_{0}^{\infty} f(Z(t) \mid Z(0)=i) \mathrm{d}{t}, \quad i >0.
\end{equation}
We assume throughout\footnote{Note that, on this finite state space, we need not assume that $f$ is nondecreasing for large arguments, as required for $(Z(t))_{t \geqslant 0}$ on $\NN_{\geq 0}$ in \cite{StefanovWang2000}.}
that $f(0)=0$; 

in particular, this implies that the integral has no contribution from beyond the time of absorption in 0.
The mathematical features of $\mathcal{I}_{i}(f)$, especially its expectation and  higher-order moments, have been well studied  \citep{puri1966homogeneous, puri1968some, jerwood1970note, mcneil1970integral, goel1974stochastic, norden1982distribution, SuarezChavez1999, StefanovWang2000, pollett2002path, pollett2003integrals, crawford2018computational, hobolth2019phase}. Here we restrict $(Z(t))_{t \geqslant 0}$ to $[N]_0$.  In Appendix \ref{sec:moment}, we outline  how to derive the moments of the path integral via first-step analysis; below we summarise the results.

For $j \in [N]_0$, let $\lambda_{j}$ and $\mu_{j}$, respectively, be the birth and death rates, that is, the rates for the transitions $j \to j+1$ and $j \to j-1$, in  $(Z(t))_{t \geqslant 0}$; we assume that $\lambda_{0}=0$ (since 0 is absorbing), $\lambda_j \geq 0$ for $j \in [N-1]$, and $\mu_j>0$ for $j \in [N]$, and complement this by the convention  $\mu_{0}= \lambda_N=0$. The expectation of $\mathcal{I}_{i}(f)$ in \eqref{def_path-integral} is then given by 
\begin{equation}\label{1st_moment_formula}
    \E{ \mathcal{I}_{i}(f)} = \sum_{j=1}^{N} \zeta_{ij} f(j),
\end{equation}
where, for $i,j\in [N]$,
\begin{equation}\label{t_ij_formula}
    \zeta_{ij} \defeq \sum_{\ell=1}^{\min\{i,j\}} \eta^{}_{\ell j},
\end{equation}
and, for $1 \leq \ell \leq j\leq N$,
\begin{equation}\label{s_ellj_formula}
    \eta^{}_{\ell j} \defeq
        \frac{\lambda_{\ell} \cdots \lambda_{j-1}}{\mu_{\ell} \cdots \mu_{j}}, 
\end{equation}
where the empty product is 1
(p.~162/163 in \citet{StefanovWang2000}; note that, for $\ell \leq j$, $\eta_{\ell j}$ equals their $H_{\ell}(j)/\mu_{\ell}$; see also \citet{SuarezChavez1999} for the special case of $i=1$, and see eq.~(1) of \citet{StefanovWang2000} for the special case $N=\infty$). 
Here, $\eta_{\ell j}$ represents the expected total sojourn time in  $j$ measured from the moment where the process reaches state $\ell$ for the first time until the moment where it reaches state $\ell-1$ for the first time; for lack of reference and convenience of the reader, we prove this fact in Appendix~\ref{sec:sojourn}. 
Consequently, $\zeta_{ij}$ is the expected total sojourn time of $(Z(t))_{t \geqslant 0}$ in   $j$ before the process is absorbed in 0, given it started  in $i$ \cite[Eq.~(4)]{StefanovWang2000}.
Generalising \eqref{1st_moment_formula}, the $m$-th moment of the path integral is given by
\begin{equation}\label{mth_moment_formula}
    \E{ \big( \mathcal{I}_{i}(f) \big)^m} = m! \sum_{j_{1}=1}^{N} \cdots \sum_{j_{m}=1}^{N} \zeta_{ij_{1}} \zeta_{j_{1} j_{2}} \cdots \zeta_{j_{m-1} j_{m}} f(j_{1}) \cdots f(j_{m}),
\end{equation}
see Appendix~B. 

\subsubsection{Moments of $L_n$}
Let us now apply the results above to our genealogies. With the choice  $(Z(t))_{t \geqslant 0} = (Y(\tau))_{\tau \geqslant 0}$, $f = \mathrm{id}$, and $i=n$,  we have $L_n = \cI_n(f)$.  With the transition rates of $(Y(\tau))_{\tau \geqslant 0}$ in \eqref{bdrates}, 
\eqref{s_ellj_formula}  evaluates to
\begin{equation}\label{s_ellj_formula_SIS}
    \eta^{}_{\ell j} = 
        \frac{1}{uj} \left(\frac{s}{Nu}\right)^{j-\ell} \frac{(N-\ell)!}{(N-j)!} = \frac{1}{uj} \left(\frac{s}{Nu}\right)^{j-\ell} \fall{(N-\ell)}{j-\ell},  \quad  1 \leq \ell \leq j \leq N,
\end{equation}
and \eqref{t_ij_formula} becomes
\begin{equation}\label{t_ij_formula_SIS}
    \zeta_{nj} = \sum_{\ell=1}^{ \min\{n,j\} } \frac{1}{uj} \left(\frac{s}{Nu}\right)^{j-\ell} \fall{(N-\ell)}{j-\ell}.
\end{equation}
In Appendix \ref{sec]branch_length}, we show that, with the help of \eqref{t_ij_formula_SIS},  \eqref{1st_moment_formula} and \eqref{mth_moment_formula} evaluate to
\begin{equation}\label{ELn_culture_withd}
    \E{L_{n}} = \sum_{m=0}^{N-1} \frac{1}{u} \left(\frac{s}{Nu}\right)^{m} \frac{\fall{N}{m+1} - \fall{(N-n)}{m+1}}{m+1}
\end{equation}
and
\begin{equation}\label{ELn2_culture_withd}
    \begin{split}
        \E{L_{n}^{2}}
        &= 2 \sum_{m_{1}=0}^{N-1} \sum_{m_{2}=0}^{N-1} \left(\frac{1}{u}\right)^{2} \left(\frac{s}{Nu}\right)^{m_{1}+m_{2}} \\ & \qquad \times \left[ \frac{\{\fall{N}{m_{1}+1} - \fall{(N-n)}{m_{1}+1}\} \fall{N}{m_{2}+1}}{(m_{1}+1)(m_{2}+1)} - \frac{\fall{N}{m_{1}+m_{2}+2} - \fall{(N-n)}{m_{1}+m_{2}+2}}{(m_{1}+m_{2}+2)(m_{2}+1)} \right]. 
    \end{split}
\end{equation}

\subsubsection{Moments of $C_n$} \label{sec:momentsCn}
It is  now  an easy exercise to translate the moments of $L_n$ into the moments of $C_n$.
The $m$th \emph{factorial moment} of a random variable $W \sim \mathrm{Poi}(\nu)$ is  $\E{\fall{W}{m}}=\nu^{m}$; this is due to the simple fact that the probability generating function of $W$ is $g(z) = e^{\nu(z-1)}$, and $\E{W^{\underline{m}}}=g^{(m)}(1)$. By \eqref{Cn to Ln}, the $m$th factorial moment of $C_{n}$ thus becomes
\begin{equation}\label{EfallCkn}
        \E{ \fall{C_{n}}{m} } 
        = \EE \big [  \EE[\fall{C_{n}}{m} \mid L_{n} ] \big ]  
        = \EE [ (\mu L_n)^m] = \mu^{m} \E{ L_{n}^{m} }, \quad m >0.
\end{equation}
Together with the identity
\[
    x^{m} = \sum_{i=0}^{m} \StirlingTwo{m}{i} \fall{x}{i},
\]
where $\StirlingTwo{m}{i}$ is a Stirling number of the second kind, 
this gives the \emph{moments} of $C_n$ in terms of the moments of $L_n$ as
\begin{equation}\label{moments_general}
    \E{ C_{n}^{m} } = \sum_{i=0}^{m} \StirlingTwo{m}{i} \E{ \fall{C_{n}}{i} } = \sum_{i=0}^{m} \StirlingTwo{m}{i} \mu^{i} \E{ L_{n}^{i} }, \quad m >0.
\end{equation}
For example, the first four moments read
\begin{equation}\label{moments_Cn}
    \begin{split}
        \E{ C_{n} } &= \mu \E{ L_{n} }, \\
        \E{ C_{n}^{2} } &= \mu \E{ L_{n} } + \mu^{2} \E{ L_{n}^{2} }, \\
        \E{ C_{n}^{3} } &= \mu \E{ L_{n} } + 3 \mu^{2} \E{ L_{n}^{2} } + \mu^{3} \E{ L_{n}^{3} }, \quad \text{and} \\
        \E{ C_{n}^{4} } &= \mu \E{ L_{n} } + 7 \mu^{2} \E{ L_{n}^{2} } + 6 \mu^{3} \E{ L_{n}^{3} } + \mu^{4} \E{ L_{n}^{4} }.
    \end{split}
\end{equation}
In particular, the variance is
\begin{equation}\label{variance_Cn}
    \begin{split}
        \VV[C_{n}]
        &= \E{ C_{n}^{2} } - (\E{ C_{n} })^{2} \\
        &= \mu \E{ L_{n} } + \mu^{2} (\E{ L_{n}^{2} } - (\E{ L_{n} })^{2}) \\
        &= \mu \E{ L_{n} } + \mu^{2} \VV[L_{n}].
    \end{split}
\end{equation}
 The latter can be seen as an instance of the standard decomposition of the variance: $\Var{C_n} = \E{\Var{C_n \mid L_n}} + \Var{\E{C_n \mid L_n}}$. So the first term, $\E{\Var{C_n \mid L_n}}=\mu \E{ L_{n} }=\E{C_n}$,  contains the variability due to the innovation process, whereas the second term, $\Var{\E{C_n \mid L_n}}=\mu^{2} \Var{L_{n}}$, comes from the variance of the  length of the genealogy. 
Inserting \eqref{ELn_culture_withd} and \eqref{ELn2_culture_withd} into  \eqref{moments_Cn} and \eqref{variance_Cn}, we obtain the first two moments of  $C_{n}$ in  closed form.

For numerical and simulation results, we  work with  $u=1$ throughout, so  one time unit corresponds to the expected life time of an individual, that is, one generation. Figure~\ref{fig:changing_s} shows the two components of $\Var{C_N}$ and their ratio as functions of $s$. First, the solid line $\mu \E{L_N} = \EE[C_N]$ illustrates the (smooth) transition from the subcritical phase with a small number of traits to the supercritical phase, where the mean number of traits increases steeply with $s$, as previously observed in similar models \citep{NakamuraWakanoAokiKobayashi2020,KobayashiWakanoOhtsuki2018,KobayashiKurokawaIshiiWakano2021}. According to the variance decomposition above, $\mu \E{L_N}$ is, at the same time, the variability of $C_N$ due to innovations. The figure  shows that,  for small values of $s$, it is the (by far) dominating contribution to $\VV[C_N]$; but the proportion decreases with $s$ and, for $s>1$, $\mu^2 \VV[L_N]$, the variability of $C_N$ due to the length of the genealogy, is the major contribution to $\VV[C_N]$. The picture is very similar for the variance components of $C_1$ (not shown).

\begin{figure}
    \centering   \includegraphics[width=80mm]{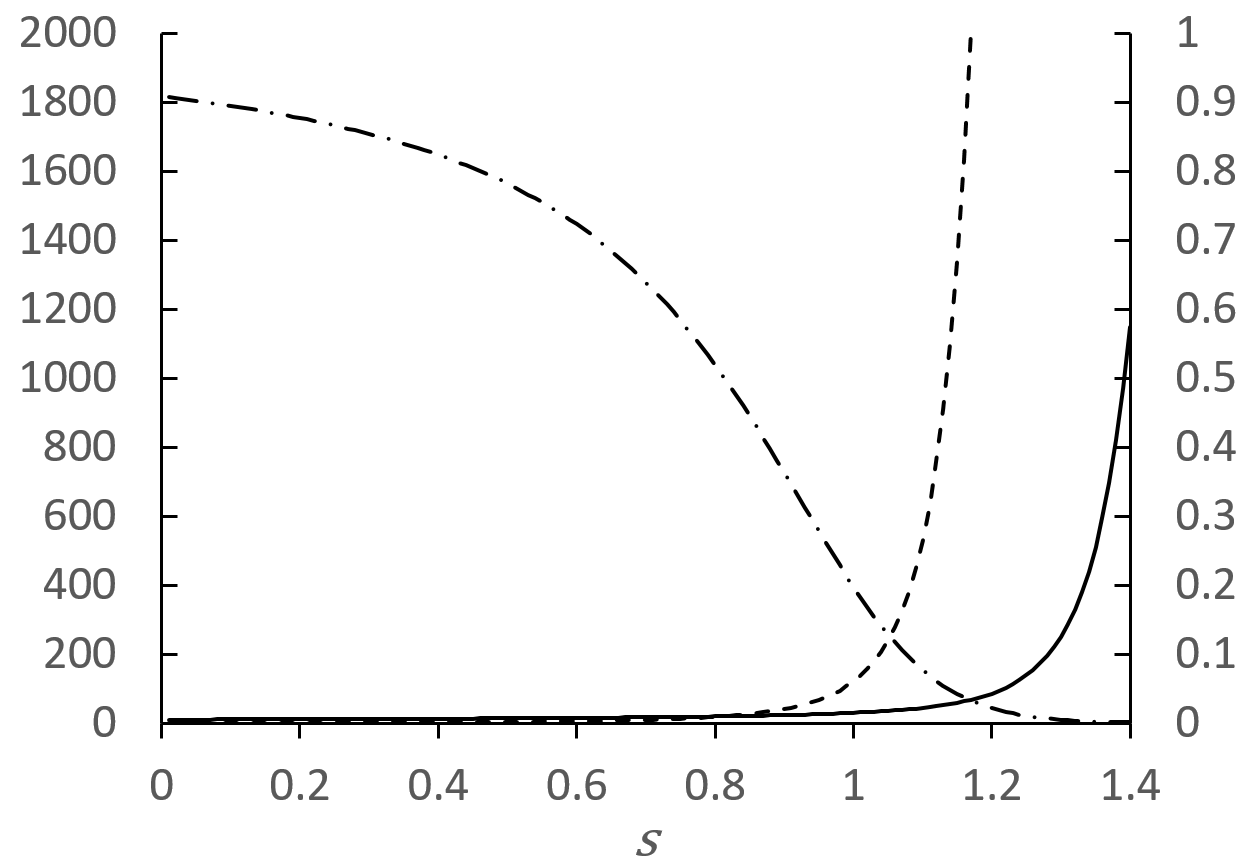}
    \caption{
    The variance components $\mu \E{L_N}$ (solid, left scale) and $\mu^2 \Var{L_N}$ (dashed, left scale), along with the proportion of the former, $\mu \E{L_N} / \Var{C_N}$
    (dot-dashed, right scale) as functions of $s$. $N=100,\mu=0.1,u=1$. 
}
   \label{fig:changing_s}
\end{figure}

Figure \ref{fig:comparison} shows $\E{C_n}$ and $\sqrt{\Var{C_n}}$ for fixed $s$ as functions of $n$ in both the sub- and the supercritical case. 
The most important difference is that, in the supercritical case, $\E{C_n}$  saturates for relatively small $n$; that is, almost all traits of the population are already contained in a small sample, a fact that will turn out as crucial in the next section.
Of course, in addition,  $\E{C_n}$ and $\sqrt{\Var{C_n}}$  are altogether substantially higher in the supercritical case, as we have already seen for $n=N$ in Figure \ref{fig:changing_s}. 

\begin{figure}
    \centering
    \includegraphics[width=60mm]{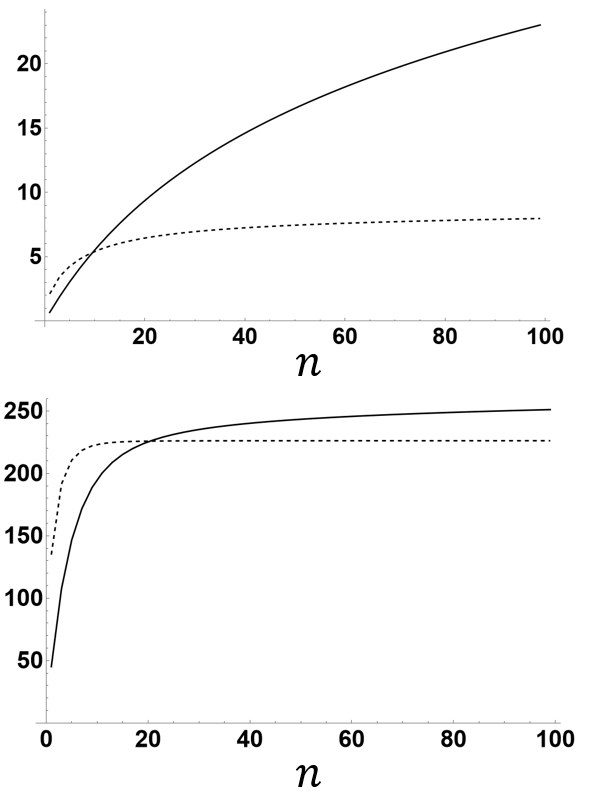}
    \caption{$\E{C_n}$ (solid) and $\sqrt{\Var{C_n}}$ (broken) as functions of $n$ in the subcritical case ($s=0.9$, upper panel) and the supercritical case ($s=1.3$, lower panel).  The graphs of $\E{L_n}$ and $\sqrt{\Var{L_n}}$ are not shown since they are almost indistinguishable from those of $\E{C_n}$ and $\sqrt{\Var{C_n}}$ except for the scaling factor of $\mu$. $N=100,u=1,\mu=0.1$.
    }
    \label{fig:comparison}
\end{figure}

\subsection{Marginal and joint distributions of  $L_1$ and $L_N$, and of $C_1$ and $C_N$}

Figure~\ref{fig:loghistogramL} shows the histograms of $\log_{10} \ell_1$ and $\log_{10} \ell_N$ (where $\ell_n$ is the realisation of $L_n$) obtained via \eqref{Ln_in_path-integral}  by simulating the untyped ALG $(\Lambda(\tau))_{\tau \geq 0}$ following  \eqref{eq:lambdaTransitions} until extinction, a large number of times.

\begin{figure}
    \centering
    \includegraphics[width=80mm]{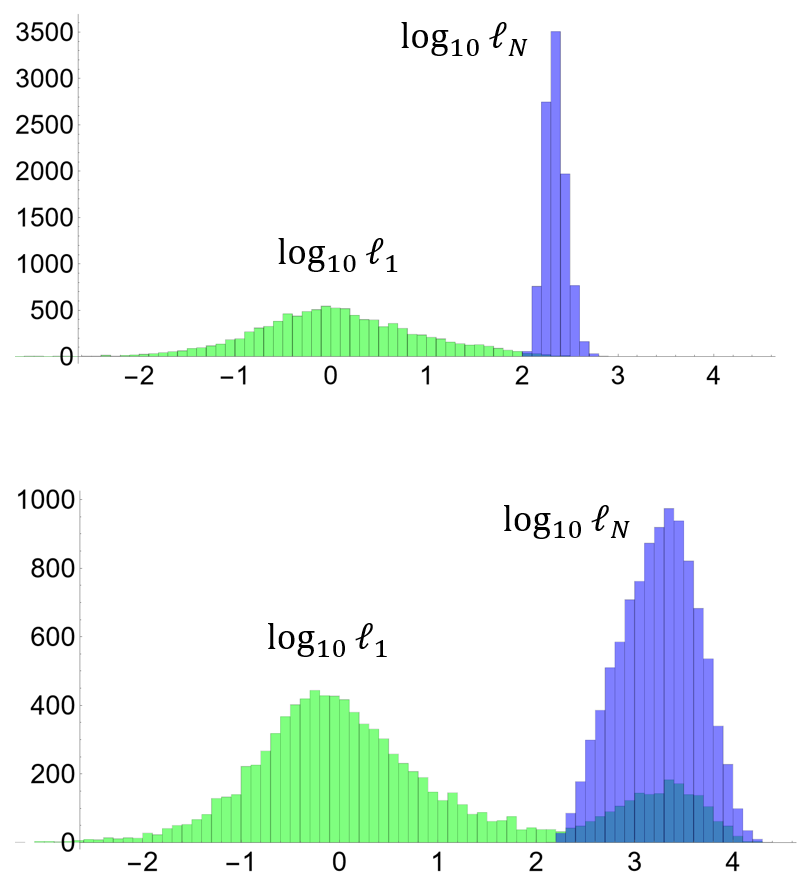}
    \caption{Histograms of $\log_{10}\ell_1$ (green) and $\log_{10}\ell_N$ (blue), obtained from $10^4$ simulation runs of the untyped ALG; the overlap area of the two histograms appears blue-green. In each run, the ancestors were traced until they disappeared. 
    Top: $s=0.9$; bottom: $s=1.3$.
    $N=100, u=1$.}
    \label{fig:loghistogramL}
\end{figure}

In the subcritical case, the histograms simply tell us that single individuals tend to have short ancestries, whereas the ancestries of the entire population are much longer (note the log scale). In  the supercritical case, the values of $\ell_N$ become much larger, and the distribution of $\ell_1$ is bimodal: it has a large peak for small values, where most of the mass is located,   and  a second smaller peak  where the distribution of $\ell_N$ is located. So most individuals have short genealogies, but occasionally they become as long as the genealogy of the entire population.

Recalling now that $C_n \sim \text{Poi}(\mu L_n)$, we turn the histograms of $\ell_1$ and $\ell_N$ into those of $c_1$ and $c_N$,  where $c_n$ is a realisation of $C_n$, see Figure \ref{fig:loghistogramC_ALG}. A qualitative difference appears only for the left peaks: the left peaks of the histograms of $\ell_1$ translate into the left peaks of the histograms of  $c_1$ in that short genealogies mean high chances for no or only a very small number of innovation events, resulting in a sharp peak of (nearly) ignorant individuals. So the histograms show that, in the supercritical case, most individuals are (almost) ignorant, but some of them carry nearly all the  knowledge of the entire population.

This becomes even clearer when we consider the joint distribution of  $\ell_1$ and $\ell_N$ in
heat maps, and likewise for $c_1$ and $c_N$ (Fig.~\ref{fig:heatmapLandC}). In the subcritical case, the correlation between $\ell_1$ and $\ell_N$ is altogether weak;  $\ell_1$ is bimodal in the rare case that $\ell_N$ is relatively large. A similar tendency exists between $\ell_2$ and $\ell_N$ (not shown).
In contrast,  the correlation is very strong in the supercritical case.  When $\ell_N$ is not too small, the distribution of $\ell_1$ is bimodal, and there are two extreme types of individuals: those with very short genealogies and those whose genealogies are nearly as long as those of the entire population; there are only very few intermediate cases. An analogous  observation applies to $\ell_2$ and $\ell_N$ (not shown).

Unsurprisingly, a very similar picture emerges in the joint distribution of $c_1$ and $c_N$. The only difference is some additional noise caused by the Poisson random variable translating $\ell_n$ into $c_n$.

\begin{figure}
    \centering
    \includegraphics[width=80mm]{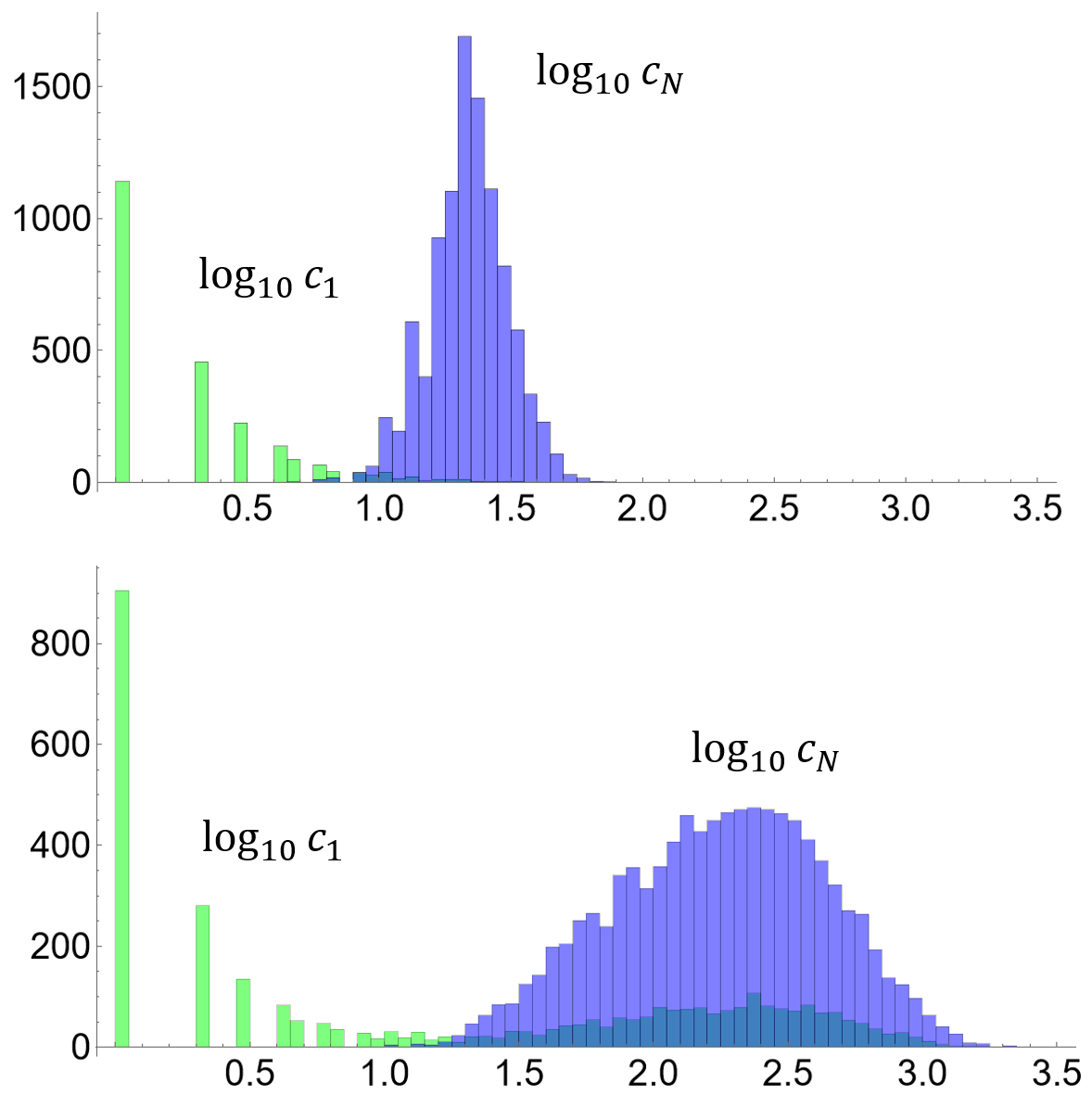}
    \caption{Histograms of $\log_{10} c_1$ (green) and $\log_{10} c_N$ (blue), where $c_1$ is a random number drawn from Poi$(\mu \ell_1)$ and likewise for $c_N$, and the values of $\ell_1$ and $\ell_N$ are those from Figure~\ref{fig:loghistogramL}.
    Top: $s=0.9$. Bottom: $s=1.3$.
    $N=100, u=1$.}
    \label{fig:loghistogramC_ALG}
\end{figure}

\begin{figure}
    \centering
    \includegraphics[width=160mm]{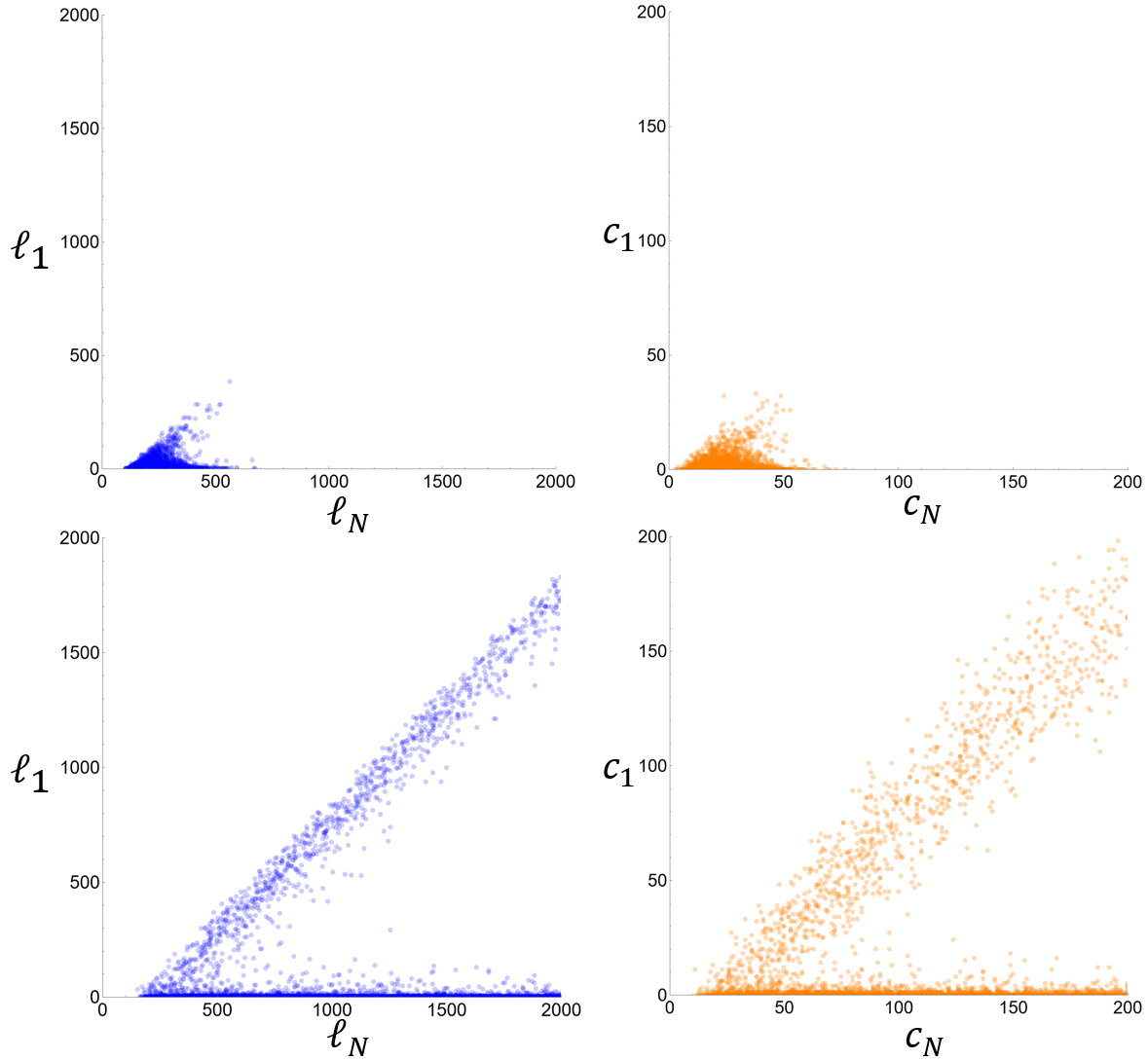}
    \caption{Heat maps showing the joint distributions of  $\ell_1$ and $\ell_N$ (blue) and of $c_1$ and $c_N$ (orange), where the values of $\ell_1, \ell_N, c_1$, and $c_N$ are those from Figures~\ref{fig:loghistogramL} and \ref{fig:loghistogramC_ALG}.
    Top: $s=0.9$, bottom: $s=1.3$.
    $N=100, u=1$.
    }
    \label{fig:heatmapLandC}
\end{figure}

\subsection{Metastability and merging in the ALG}
\label{subsec:metastability}
Still, we are at the descriptive level. For an understanding of the behaviour, we must consider the \emph{dynamics} of the ALG. As mentioned above, the line-counting process $(Y_n(\tau))_{\tau \geq 0}$ dies out with probability 1; but it is decisive what happens before extinction. Let us briefly recapitulate the essentials following \citet{AnderssonDjehiche1998} and \citet{Foxall2021}. In the subcritical case, $(Y_n(\tau))_{\tau \geq 0}$ dies out quickly almost surely. In contrast, in the supercritical regime, it has a metastable state around $N\bar \xi$, where  $\bar \xi$ is the stable equilibrium of the differential equation \eqref{ode_xi}. Note  that $\lambda_{\lfloor N \bar \xi\rfloor} \approx \mu_{\lfloor N \bar \xi\rfloor}$; also recall from Section~\ref{subsec:LLN} that the qualitatively different behaviour in the sub- and supercritical regimes reflects the transcritical bifurcation of the equilibria of the differential equation.

For $(Y_1(t))_{\tau \geq 0}$ in the supercritical case, there is the following dichotomy: with probability $u/s$, $(Y_1(\tau))_{\tau \geq 0}$ goes extinct quickly; otherwise, it grows to reach the metastable state (say $\lfloor N \bar \xi\rfloor$ for definiteness) in a short time. $(Y_N(\tau))_{\tau \geq 0}$ always moves to $\lfloor N \bar \xi\rfloor$ quickly. By a  crude approximation, the expected first-passage times from 1 to $\lfloor N \bar \xi\rfloor$ (conditional on non-extinction before reaching $\lfloor N \bar \xi\rfloor$), and from $N$ to $\lfloor N \bar \xi\rfloor$, are both bounded by $N^2/u$ \cite[Proof of Lemma 3]{AnderssonDjehiche1998}\footnote{The factor of $1/u$ comes from the fact that \cite{AnderssonDjehiche1998} work with $u=1$.}. (According to \cite{AnderssonDjehiche1998}, it can  be shown with the help of refined arguments that the bound is actually $\cO(\log(N))$.)

Whenever the process has reached the metastable state, it fluctuates around it for a long time with fluctuations of order $\sqrt{N}$, before finally going extinct in a rare and rather sudden event. See the comprehensive work of \citet{Foxall2021} for the details. The expectation of the extinction time $T$ when starting from the metastable state, or, more generally, from any initial value $\cO(N)$, in a population of size $N$ is given by
\begin{equation}\label{ET}
    \E{T} 
    = \frac{s}{(s-u)^2} \sqrt{\frac{2 \pi}{N} } e^{N \left\{ \log (s/u) + (u/s) -1 \right\} }  (1+\scO(1)), 
\end{equation}
and the distribution of $T/\E{T}$ converges to an exponential distribution with parameter 1,
see  \citet{AnderssonDjehiche1998,doering2005}; \citet[eq. (12.2)]{Nasell2011}, and \citet{Foxall2021}. Since $\log(x) + 1/x \geq 1$ for $x \geq 1$ with  equality if and only if $x=1$ (note that $\log(1)+(1/1)-1=0$ and $(\dd/\dd x) (\log(x)+(1/x) -1)=(1/x) - (1/x^2) > 0$ for $x>1$), \eqref{ET} means that the time the process spends in the metastable regime increases exponentially with $N$. Our previous attributes such as `quickly' and `in a short time' are actually meant relative to this exponential time scale.

But the line-counting process alone does not suffice to understand what is really going on; rather, one has to consider the dynamics of the ALG, more precisely, the processes $(\Lambda_\alpha(\tau))_{\tau \geq 0}$ for $\alpha \in [N]$ and $(\Lambda_{[N]}(\tau))_{\tau \geq 0}$ of \eqref{eq:lambdaTransitions}. In line with the dichotomy described above, the simulations in Figure \ref{fig:numAncestors} show that the ancestors of a single individual either die out quickly or grow to a metastable number and survive for a long time; likewise, the number of ancestors of the entire population reaches the metastable regime quickly. But we also see that, soon after both $(\Lambda_\alpha(\tau))_{\tau \geq 0}$ and $(\Lambda_{[N]}(\tau))_{\tau \geq 0}$ have reached a metastable size,  \emph{they become identical}.  This is  due to the fact that $\Lambda_\alpha(\tau) \subseteq \Lambda_{[N]}(\tau)$ for all $\tau \geq 0$, and the sizes of the two sets  perform fluctuations around $N \bar \xi$ of order $\cO(1/\sqrt{N})$, in the course of which they finally meet and hence become identical. From this point onwards, we have $\Lambda_\alpha(\tau)=\Lambda_{[N]}(\tau)$ due to \eqref{identical_sets}.

The merging events\footnote{Note that these merging events  differ from the coalescence events in the ASG, where  two lines unite into one.} responsible for this happen whenever, for some $\alpha$, a line in $\Lambda_\alpha(\tau)$ learns from a line in $\Lambda_{[N]}(\tau)\setminus \Lambda_\alpha(\tau)$: then, a parent from $\Lambda_{[N]}(\tau)\setminus \Lambda_\alpha(\tau)$ is added to $\Lambda_\alpha(\tau)$ and hence removed from $\Lambda_{[N]}(\tau)\setminus \Lambda_\alpha(\tau)$. A transition in the reverse direction (that is, an addition of elements to $\Lambda_{[N]}(\tau)\setminus \Lambda_\alpha(\tau)$) is not possible: when a line in $\Lambda_{[N]}(\tau)\setminus \Lambda_\alpha(\tau)$ learns from a line in $\Lambda_\alpha(\tau)$, neither $\Lambda_\alpha(\tau)$ nor $\Lambda_{[N]}(\tau)$  will change.

\begin{figure}
    \centering    \includegraphics[width=160mm]{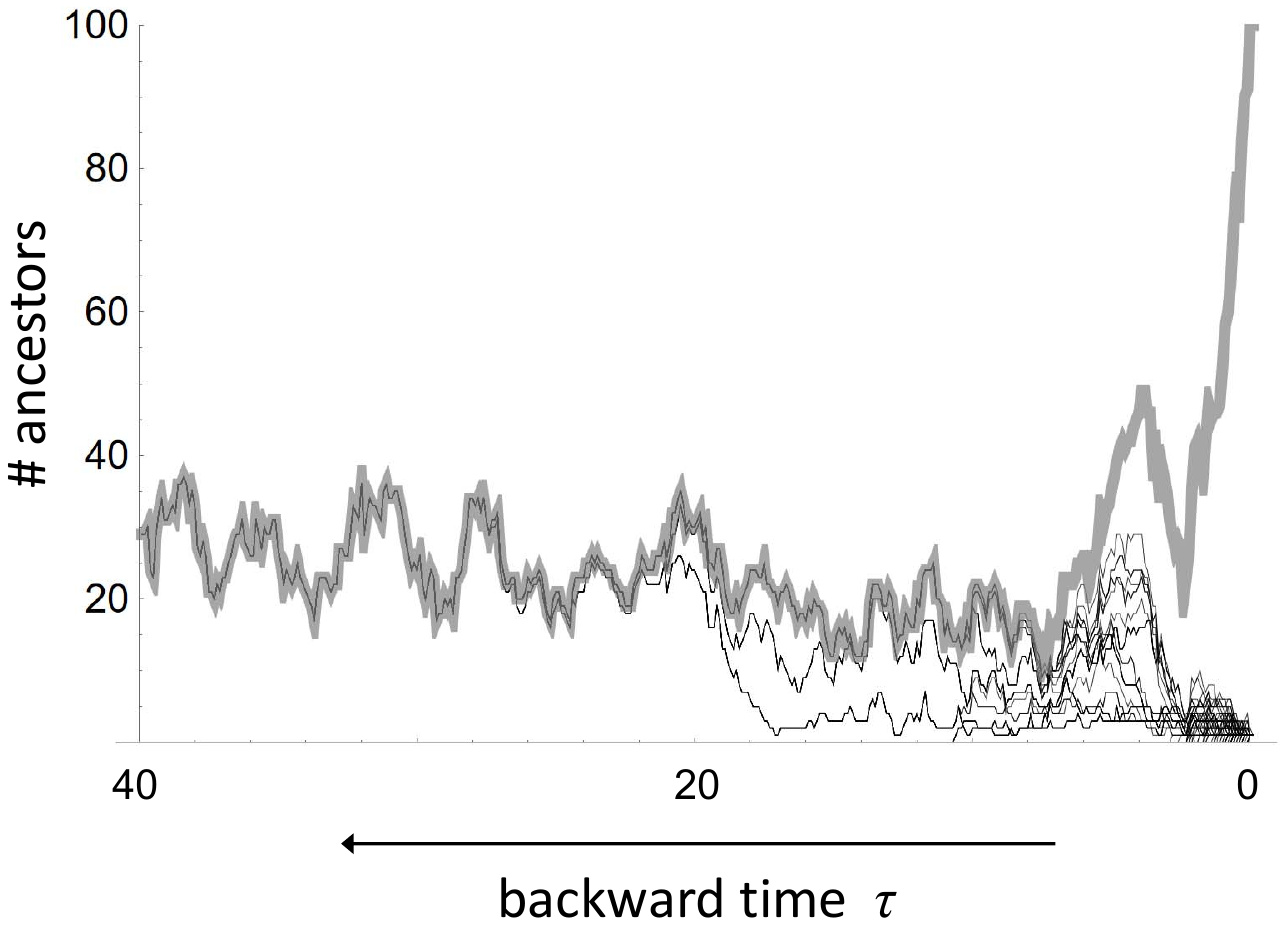}
    \caption{Simulation of the untyped ALG. $\lvert (\Lambda_{[N]}(\tau))_{\tau \geq 0}\rvert$, the number of ancestors of the total population over time (thick grey line), and the $\lvert (\Lambda_\alpha(\tau))_{\tau \geq 0}\rvert$, the sizes of the ancestries of each individual $\alpha$ (100 thin black lines for $\alpha \in [N]$) are shown from backward time $\tau$=0 to $\tau=40$, all based on the same realisation of $(\Lambda(\tau))_{\tau \geq 0}$. By definition, the thick grey line starts at $N$, and each thin black line starts at 1.  When a black line merges into the grey line (26 occurrences among 100, which is close to the probability of $1-u/s=0.23$), the corresponding sets become identical and behave identically from then on. $N=100, s=1.3, u=1$.
    }
    \label{fig:numAncestors}
\end{figure}

Nevertheless, this process is not one-dimensional:  there are also death events in $\Lambda_\alpha(\tau)$ and $\Lambda_{[N]}(\tau) \setminus \Lambda_\alpha(\tau)$, as well as learning events from parents outside $\Lambda_N(\tau)$, by which either of the two sets increases or decreases individually (these types of events are responsible for the fluctuations around the metastable state).

All this explains the stationary behaviour illustrated in Figure~\ref{fig:heatmapLandC}. In the subcritical case, the ancestries of all individuals die out quickly without much of a chance to merge, so most of them remain short and more or less independent.  This changes in the supercritical case.  If the ancestry of a given individual dies out quickly, it will again remain short and independent of $L_N$; this corresponds to the concentration of points close to the horizontal axis in the  corresponding heat map. But if the number of ancestors moves to the metastable state, it will stay there for a long time, so the individual has old ancestors. Moreover, due to (relatively fast) merging, the set of old ancestors, and thus the set of traits, is largely shared with the entire population. This corresponds to the cloud of points close to the diagonal. 
As a consequence, cultural traits that are old but  not too old to have gone extinct are carried by all knowledgeable individuals. So, for old traits that still exist in a population, there are  two main types of individuals: those who know none of them and those who know all of them. Altogether, the trait diversity between individuals is low.

Let us mention that a  fast merging of  ancestral sets  similar to that observed above was described, and actually proved, for the minimal-load ASGs in a model of Muller's ratchet with tournament selection \cite[Sec.~7]{GCSmadiWakolbinger23}. Like the ALG, the minimal-load ASG experiences branching and pruning events; but it additionally has coalescence events, which  are crucial for the proof, so the latter does not carry over to our situation.


\subsection{Forward dynamics and evolving genealogies}\label{subsec:simulations}

So far, we have considered the stationary distributions of $L_n$ and $C_n$. Let us  now return to the dynamics $(L_n^t)_{t \geq 0}$ and $(C_n(t))_{t \geq 0}$; recall that $(L_n^t)_{t \geq 0}$ is the length of the evolving genealogy, and the upper index refers to the time where we start looking back.


\subsubsection{Forward simulations}
We carried out individual-based simulations of the forward process $(\Phi(t))_{t \in [0,t_{\rm{max}}]}$  described in Section \ref{subsec:FullForward}.  
We start with $k_i(0)=\varnothing$ for every $i\in [N]$ and compute realisations $c_n(t)$ of $C_n(t)$ at $t=0,1,...,t_{\rm{max}}$. We take  $t_{\rm{max}}$ large enough so that the effect of the initial state can be neglected and denote by $[t_{\rm{max}}]_0$ the set $\{0,1,\ldots,t_{\rm{max}}\}$. 


\subsubsection{Evolving genealogies}\label{sec:EvolvingGenealogies}

We are now ready to consider $L_n^t$, the length of the genealogy as a function of the forward time $t$ where we start looking back, as defined in \eqref{Ln_in_path-integral}; and we understand it as the process $(L_n^t)_{t \geq 0}$   
coupled across all times via the graphical representation, as illustrated in Figure~\ref{fig:Lnt_ver2}. The same applies, of course, to the $(\Lambda^{t}_{[n]}(\tau))_{\tau \geq 0}$.   We clearly have
\begin{equation}
\Lambda^t_{[n]}(0) = [n].
\end{equation}
As before, the set $[n]$ is representative of any sample of size $n$ due to exchangeability. The coupling across times extends to the extinction times  $T^t$  of $(\Lambda^{t}_{[N]}(\tau))_{\tau \geq 0}$;  the $T^t$ are all identically distributed in the same way as $T$ of \eqref{ET}, but they are, in general, not  independent. 

All this leads us to the concept of \emph{evolving genealogies}, as previously studied in the context of Kingman's coalescent process (that is, with coalescence, but without branching events) by \cite{PfaffelhuberWakolbinger2006} and \cite{PfaffelhuberWakolbinger2011}.  In practice, we have extracted the evolving genealogies from a single long simulation run of the forward model, where we have stored  all events (including  the information about who died and who learned from whom) and the times at which they happen.  This way, we obtained  realisations of
$(\Lambda^t_{n})_{t \geq 0}$ and $(L^t_{n})_{t \geq 0}$ coupled across all times. We will now discuss them together with the realisations of $(C_{n}(t))_{t \geq 0}$ extracted directly from the forward simulation.

\begin{figure}
    \centering    \includegraphics[width = 140 mm]{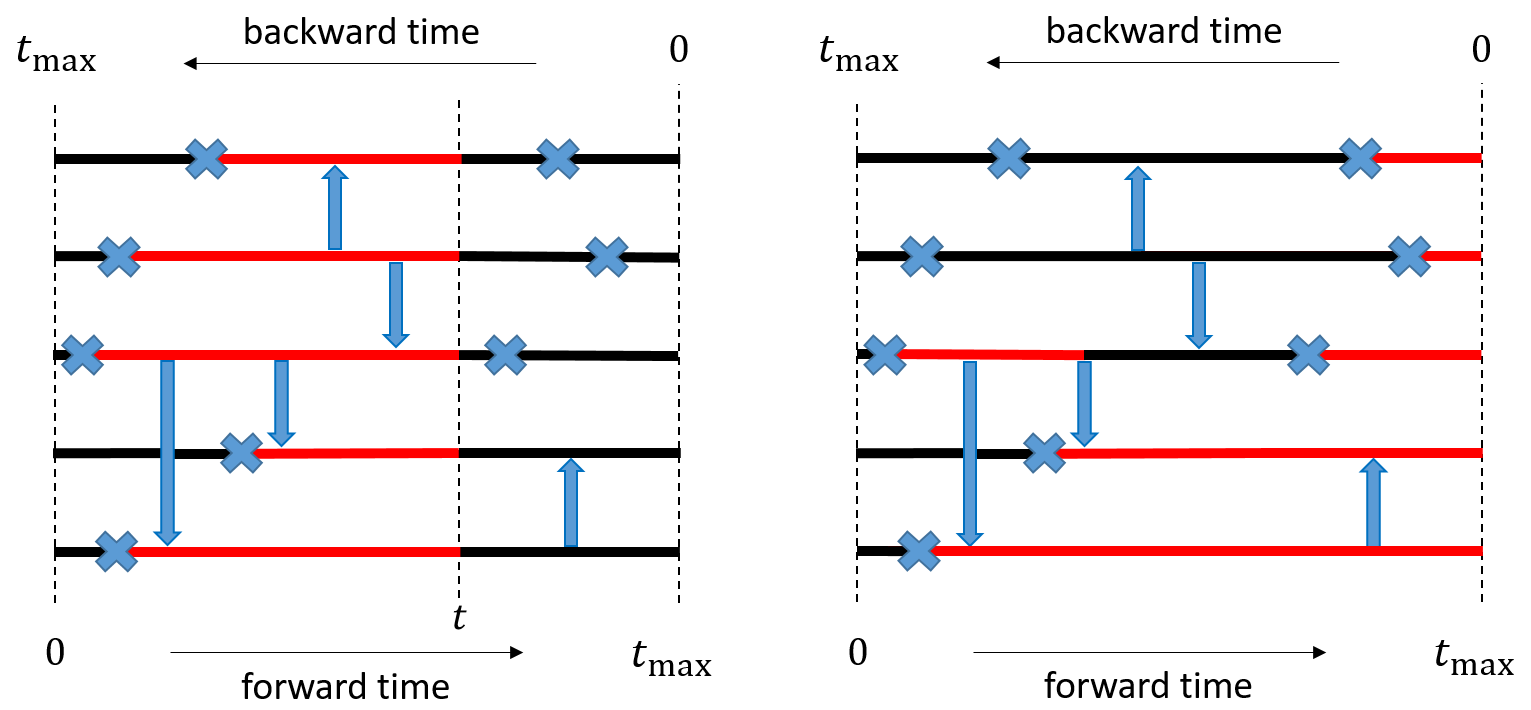}
    \caption{The total length $L_N^t$ of the  genealogy  of the entire population as a function of forward time $t$. The realisation of the graphical representation, as given by the positions and kinds of all events, is identical in both panels.  The top (bottom) panel shows a case with $0<t<t_{\rm{max}}$ (with $t=t_{\rm{max}}$). The sum of the lengths of the red line segments is  $L_N^t$.
    }
    \label{fig:Lnt_ver2}
\end{figure}

\subsubsection{Forward dynamics: length of genealogy and number of traits}

Long time series $(\mu \ell_N(t))_{t \geq 0}$, $(c_N(t))_{t \geq 0}$, $(\mu \ell_1(t))_{t \geq 0}$, and $(c_1(t))_{t \geq 0}$ are shown in Figures \ref{fig:FullTimeSeries_sub} and \ref{fig:FullTimeSeries_super} for the  subcritical   and the supercritical case, respectively.
More detail is to be seen in Figure \ref{fig:SampleTimeSeries}, which shows short cutouts of $(\ell_N(t))_{t \geq 0}$ and $(\ell_1(t))_{t \geq 0}$  in both cases. In the subcritical case, both $(\ell_1(t))_{t \geq 0}$ and $(\ell_N(t))_{t \geq 0}$ fluctuate moderately around some constant value in a more or less symmetric manner, and likewise for the traits. The number of traits seems to be largely independent of the length of the genealogies.

In contrast, in the supercritical case, the  time series $(\mu \ell_N^t)_{t \geq 0}$ and $(c_N(t))_{t \geq 0}$ are very close in our simulations, and likewise for $(\mu \ell_1^t)_{t \geq 0}$ and $(c_1(t))_{t \geq 0}$. 
Before we turn to the conspicuous shape of the supercritical dynamics, let us comment on the different degrees of coupling of $(\ell_n^t)_{t \geq 0}$ and $(c_n(t))_{t \geq 0}$ (for $n \in \{1,N\}$) in the sub- and supercritical cases. It reflects the variance components of $C_n$ in  Section~\ref{sec:momentsCn} and Figure~\ref{fig:changing_s}: recall that  the variance due to the innovation process is an appreciable proportion of $\VV[C_n]$ in the subcritical case, but is negligible for the supercritical case, so $\VV[C_n]$ is then largely governed by the variability of the length of the genealogy. This explains that $(\ell_n^t)_{t \geq 0}$ and $(c_n(t))_{t \geq 0}$ fluctuate more or less independently for $s<1$, but are closely coupled for $s>1$.

As to the dynamics in the supercritical case, both $(\ell_N(t))_{t \in [t \geq 0}$  and $(c_N(t))_{t \geq 0}$ display a characteristic sawtooth behaviour: they  increase more or less linearly (at roughly constant slope in every sawtooth) for some period, followed by a rapid collapse to a small value at a random time. We will find the reason for this in the next section. 

The dynamics of $(\ell_1(t))_{t \geq 0}$  and $(c_1(t))_{t \geq 0}$ in the supercritical case confirm the dichotomy observed in Section~\ref{subsec:metastability}: individual ancestries quickly either die out or move to the metastable state, so they are either very short or nearly as long as the genealogy of the entire population; as a consequence, individuals know `nearly nothing' or `nearly everything'. The dynamics yields the additional insight that the role of any given individual  (short or long genealogy) changes quickly, typically many times within a given sawtooth, see the lower panel of Figure \ref{fig:FullTimeSeries_super}; this happens when an individual with a short history is attached to an individual with a long history via a learning arrow; or, the other way around, if an individual with a long history dies.

As a consistency check, we have also produced the histograms of $(\ell_1(t))_{t \in [t_{\text{max}}]_0}$,  $(c_1(t))_{t \in [t_{\text{max}}]_0}$,  $(\ell_N(t))_{t \in [t_{\text{max}}]_0}$, and $(c_N(t))_{t \in [t_{\text{max}}]_0}$ resulting from our time series, and they are practically indistinguishable from the histograms in Figures~\ref{fig:loghistogramL} and \ref{fig:loghistogramC_ALG}, as it should be by stationarity and ergodicity.

\begin{figure}
 \centering    \includegraphics[width = 160mm]{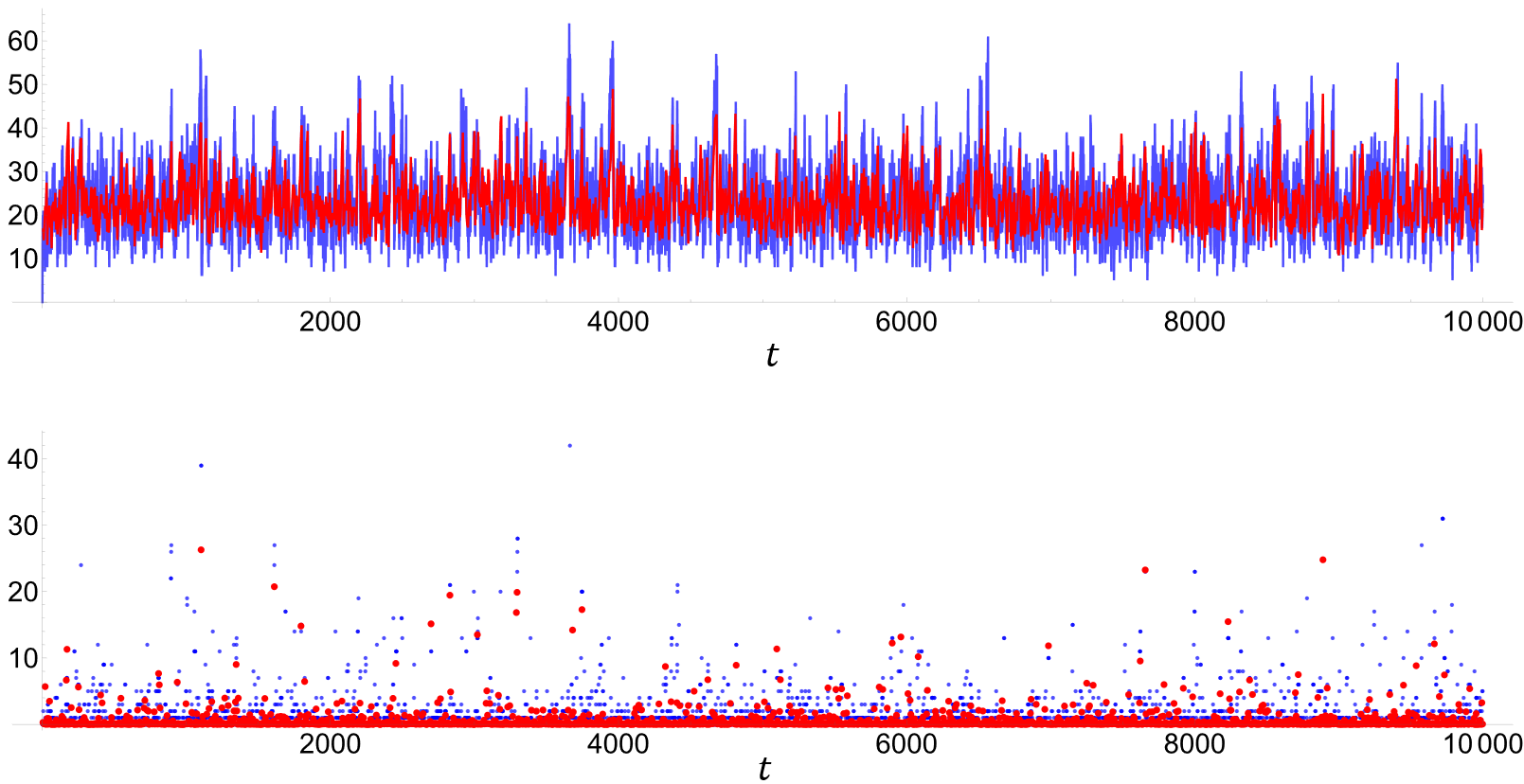}
 \caption{\label{fig:FullTimeSeries_sub}
 The innovation rate times the length of the genealogy (red)  and the number of traits (blue) as functions of forward time $t$ in the same realisation of a forward simulation. Top: $\mu (\ell_N^t)_{t \geq 0}$  and  $(c_N(t))_{t \geq 0}$. Bottom: $\mu (\ell_1^t)_{t \geq 0}$  and $(c_1(t))_{t \geq 0}$. In the bottom panel,  both time series oscillate quickly between near-0 and close to $\mu \ell_N^t$, so we did not interpolate between the points.
 $N=100, s=0.9, \mu=0.1, u=1$.}
\end{figure}

\begin{figure}
    \centering    \includegraphics[width=160mm]{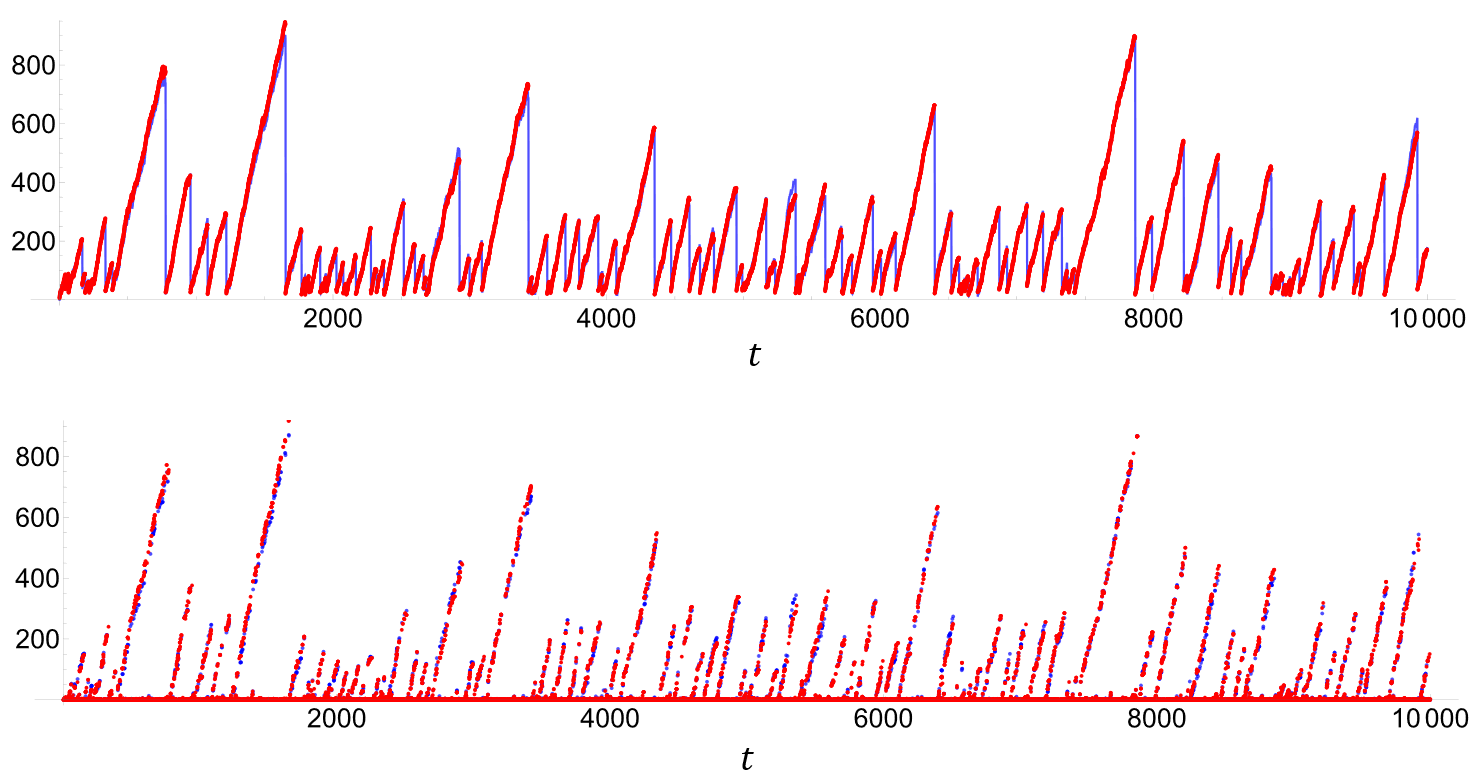}    \caption{\label{fig:FullTimeSeries_super}
    The innovation rate times the total length of the genealogy (red)  and the number of traits (blue) as functions of forward time $t$ in the same realisation of a forward simulation. Top: $\mu (\ell_N^t)_{t \geq 0}$ and  $(c_N(t))_{t \geq 0}$. To make the blue curve visible `below' the red, we did not interpolate between the red points. Bottom: $\mu (\ell_1^t)_{t \geq 0}$ and $(c_1(t))_{t \geq 0}$. Here, both time series oscillate quickly between near-0 and close to $\mu \ell_N^t$, so we did not interpolate between any points.
    $N=100, s=1.3, \mu=0.1, u=1$.
    } 
\end{figure}

\begin{figure}
    \centering
    \includegraphics[width=80mm]{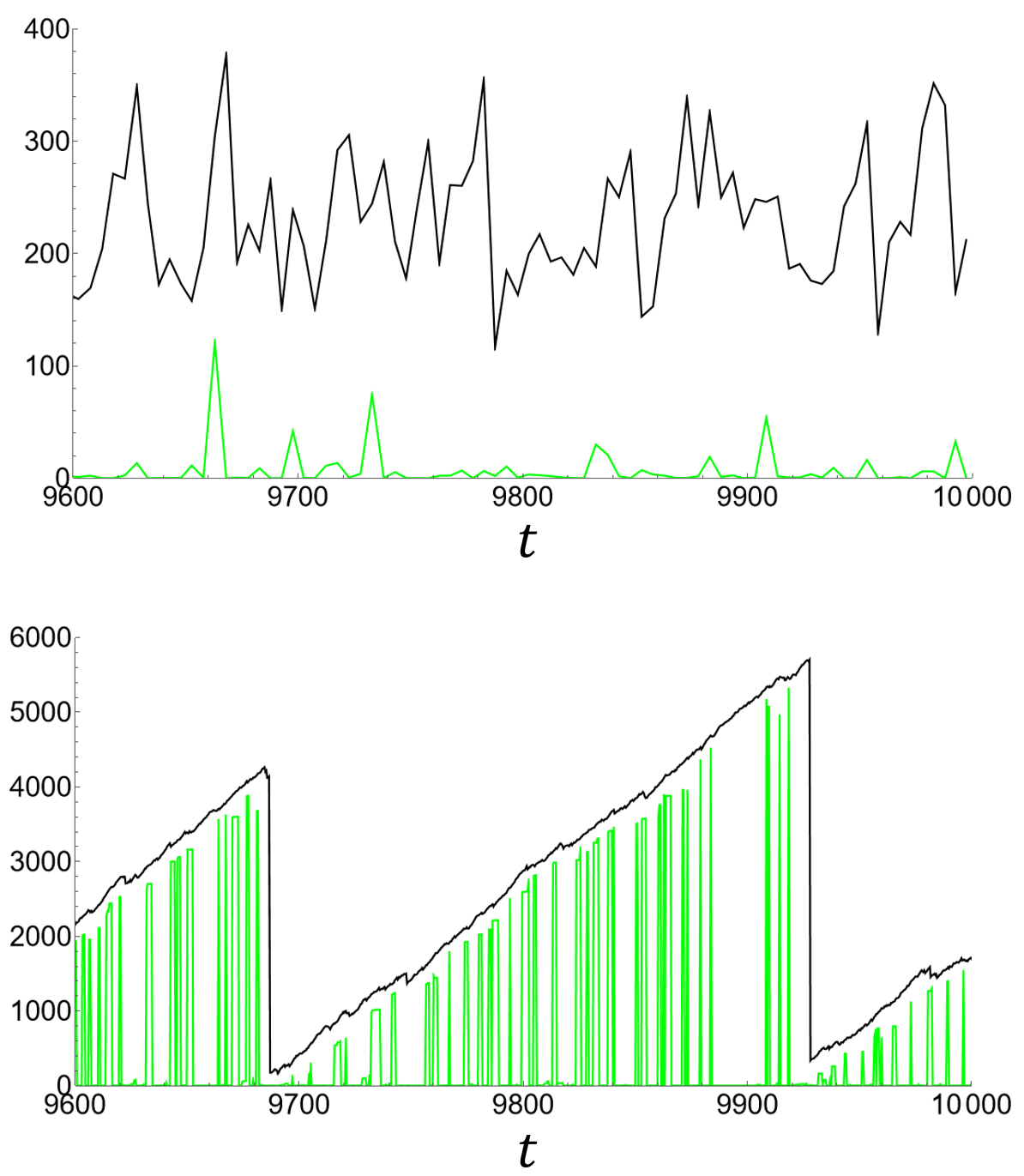}
    \caption{Time series $\big (\ell_1(t) \big )_{t \geq 0}$ (green) and $\big (\ell_N(t) \big )_{t \geq 0}$ (black) for the last 400 time units in a run with  $t_{\rm{max}}=10^4$ in the subcritical case (upper panel, $s=0.9$) and the supercritical case (lower panel, $s=1.3$).
    $N=100,\mu=0.1,u=1$.
    }
    \label{fig:SampleTimeSeries}
\end{figure}

\begin{figure}
    \centering    \includegraphics[width = \textwidth]{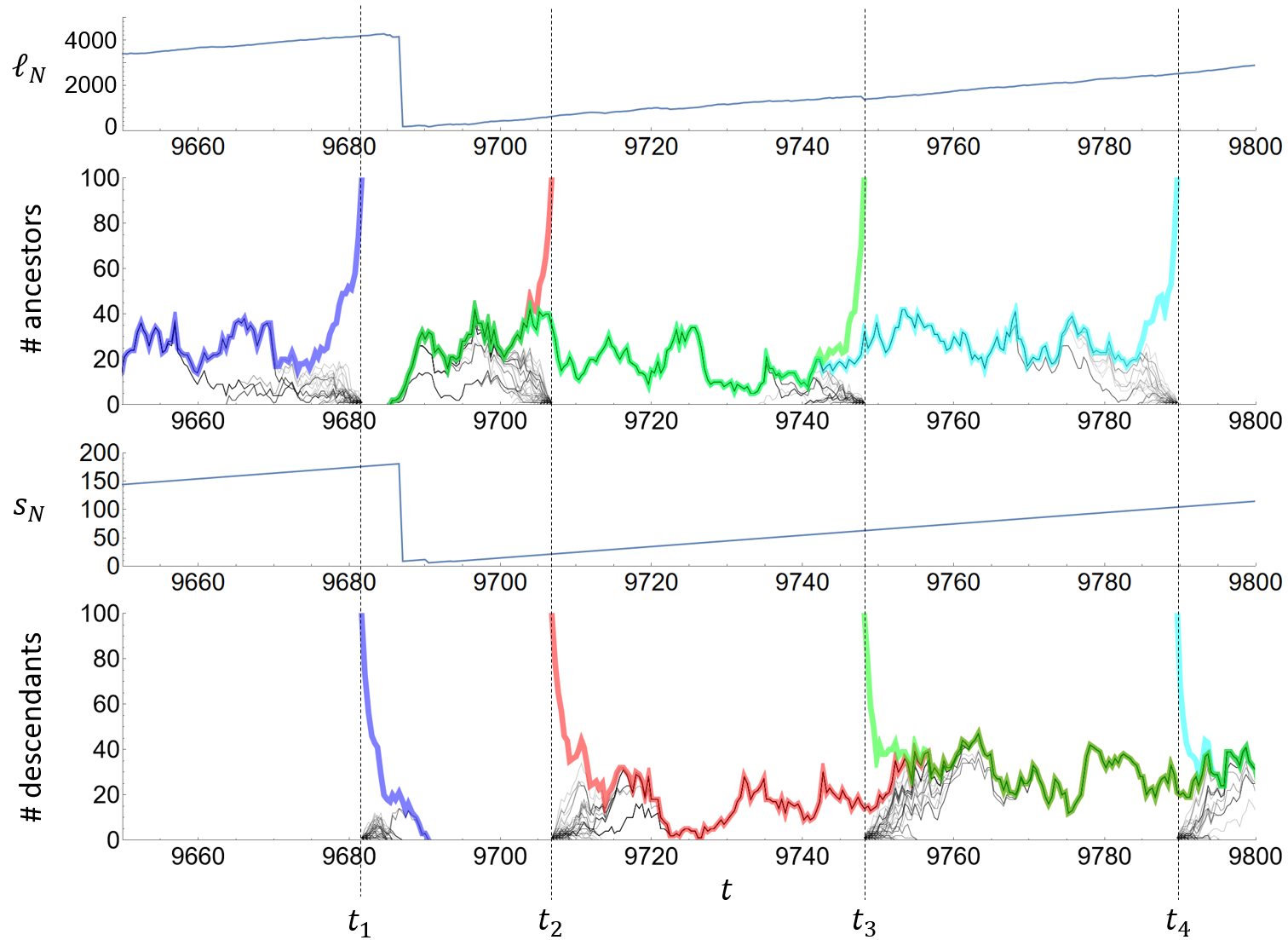}
    \caption{Relationship between (from top to bottom) the total  length of the genealogy, the set of ancestors, the time to the oldest ancestor, and the descendant process, all extracted from the same forward simulation, and all as a function of forward time $t$. In detail: the top panel shows a realisation $(\ell_N^t)_{t \geq 0}$ of $(L_N^t)_{t \geq 0}$ with a mass extinction (that is, an abrupt collapse to near zero) at $t\approx 9688$. The second panel displays the number of ancestors at $t \leq t_i$ ($i=1,2,3,4)$ of samples taken at forward times $t_1=9681, t_2=9706, t_3=9748$, and $t_4=9789$. More precisely, the thick coloured lines are realisations of the $\lvert \Lambda^{t_i}_{[N]}(t_i-t)\rvert_{t \in [t_i]_0}$, whereas the thin grey lines represent realisations of $\lvert \Lambda^{t_i}_{\alpha}(t_i-t)\rvert_{t \in [t_i]_0}$ for all $\alpha \in [N]$; these lines get darker when more lines overlap.
    The third panel shows the realisation $(s_N(t))_{t \geq 0}$ of $(S_N^t)_{t \geq 0}$, the time to the oldest ancestor.
    The bottom panel displays the number of descendants at time $t \geq t_i$  of samples taken at the forward times $t_i$. In analogy with the second panel, the thick coloured lines and the thin grey lines are realisations of the $\lvert \Gamma^{t_i}_{[N]}(t -t_i)\rvert_{t \geq t_i}$ and the  $\lvert \Gamma^{t_i}_{\alpha}(t-t_i)\rvert_{t \geq t_i}$, respectively. Looking backward (ancestral process) or forward (descendant process) in time, the thick coloured lines (along with the thin black ones that have merged into them) merge and fluctuate around their metastable value for a long time before going extinct. $N=100, u=1, s=1.3.$
    }
    \label{fig:numAnDe}
\end{figure}

\subsection{Understanding the sawtooth: ancestral sets and descendant process}

To understand the sawtooth behaviour observed in the supercritical case, that is, the  linear increase of $(L_N^t)_{t \geq 0}$ interrupted by sudden near-extinctions at  random times, we need  the finer picture of $\lvert \Lambda^t_{[n]}(\tau)\rvert_{\tau \geq 0}$ in Figure~\ref{fig:numAnDe}. 
The figure shows realisations of $\lvert \Lambda^{t_i}_{[N]}(t_i-t)\rvert_{t \in [t_i]_0}$ and $\lvert \Lambda^{t_i}_{\alpha}(t_i-t)\rvert_{t \in [t_i]_0}$ for $\alpha \in [N]$,
$i \in \{1,2,3,4\}$, and starting times $t_i$ chosen so that
\begin{equation}\label{t1234}
  t_1 \lesssim t_2^{\rm e} = t_3^{\rm e} = t_4^{\rm e} \ll t_2 \ll t_3 \ll t_4,
\end{equation} 
where $t_i^{\rm e}$ is the  time point where $( \Lambda^{t_i}_{[N]}(t_i-t))_{t \in [t_i]_0}$ 
goes extinct; so $t_i - t_i^{\rm e}$ is the corresponding realisation of $T^{t_i}$. 
Note that $T^{t_{2}},T^{t_{3}}$, and 
$T^{t_{4}}$ are coupled. With $\lesssim$ and $\ll$, we indicate that the quantities are close and not close to each other, respectively. 
We now see that the realisations of $(\Lambda^{t_3}_{[N]}(t_3-t))_{t \in [t_3]_0}$ and $(\Lambda^{t_2}_{[N]}(t_2-t))_{t \in [t_2]_0}$ quickly get close to, and actually become identical with, the realisation of $(\Lambda^{t_4}_{[N]}(t_4-t))_{t \in [t_4]_0}$; in particular, they join into the same metastable set and are extinguished with it. Likewise, the realisations of
$(\Lambda^{t_i}_{\alpha}(t_i-t))_{t \in [t_i]_0}$, $i \in \{2,3,4\}$ and $\alpha \in [N]$, either die out quickly or merge with the metastable set. This merging is clear by   arguments similar to those in Section~\ref{subsec:metastability}: since  $\Lambda^{t_4}_{[N]}(t_4-t_3) \subseteq \Lambda^{t_3}_{[N]}(0) =[N]$, we have $\Lambda^{t_4}_{[N]}(t_4-t_3+\tau) \subseteq \Lambda^{t_3}_{[N]}(\tau)$ for all $\tau>0$; so, by moving towards the metastable size, the latter process joins into  the former from above. Likewise, $\Lambda^{t_3}_{\alpha}(0) \subseteq  \Lambda^{t_3}_{[N]}(0)$, so $\Lambda^{t_3}_{\alpha}(\tau) \subseteq  \Lambda^{t_3}_{[N]}(\tau)$ for $\tau >0$; hence the former process, if it does not die out quickly, joins into the latter from below. Analogous arguments hold for the other time points.

Now, since $\lvert\Lambda^{t}_{[N]}(\tau)\rvert_{\tau \geq 0}$ spends most of its time alive near $N \bar \xi$,  \eqref{Ln_in_path-integral} tells us that   $L_N^t \approx N \bar \xi T^t$. In our specific realisation, we have $\ell_N^{t} \approx N \bar \xi (t - t_4^{\rm e})$ for any time $t_4^{\rm e} \lesssim t \lesssim t_4$ due to the identity of the $t_i^{\rm e}$, so we see a linear decrease (at rate $\approx N \bar \xi = N (1-u/s)$) with decreasing $t$. Restarting at $t_1$, in contrast, leads to a new metastable state, which awaits its extinction at $t_1^{\rm e}$, independently of $t_2^{\rm e}$. Moreover, since $t_1$ is close to $t_2^{\rm e}$, the value of $\ell_N^{t_1}$ is close to the maximum of the previous sawtooth. More generally, we are led to conjecture that $L_N^t$ decreases approximately linearly at rate $N \bar \xi$ for decreasing $t$ as long as $T^t \gg 0$ and then, after surpassing some small value, moves quickly to a new peak.

So far, we have followed the ancestral process backward in time. But the figure can also be read forward. Let $\Gamma_{[n]}^t(\varrho)$ be the set of descendants at time $t+\varrho$, $\varrho\geq 0$, of the set $[n]$ of individuals at forward time $t$. The process $(\Gamma_{[n]}^t(\varrho))_{\varrho \geq 0}$ has the same law as $(\Lambda_{[n]}^t(\tau))_{\tau \geq 0}$ with $\tau$ replaced by $\varrho$; this follows immediately from the fact that the set of descendants (ancestors) is obtained by following all learning arrows in the forward (backward) direction and pruning lines that meet a cross, and a reversal of all arrows does not change the law of the process.  As a side remark, let us note that  $\lvert \Gamma_{[n]}^0(t) \rvert_{t \geq 0}$ has the same law as $(X(t))_{t \geq 0}$ of our single-trait model with $X(0)=n$; this  reflects the self duality of the SIS model. In any case, like $\lvert\Lambda_{[N]}^t(\tau)\rvert_{\tau \geq 0}$,  also $\lvert\Gamma_{[N]}^t(\varrho)\rvert_{\varrho \geq 0}$ moves to the metastable state around $N \bar \xi$ quickly, and  like $\lvert\Lambda_{\alpha}^t(\tau)\rvert_{\tau \geq 0}$,  also $\lvert\Gamma_{\alpha}^t(\varrho)\rvert_{\varrho \geq 0}$ dies out quickly with probability $u/s$ and otherwise moves to the metastable state.

The important point now is that the two processes are connected via
\[
\alpha \in \Lambda_\beta^t(t-r) \Longleftrightarrow \beta \in \Gamma_\alpha^{r}(t-r), \quad \alpha, \beta \in [N], \; 0 \leq r \leq t,
\]
because if an individual $\alpha$ is ancestral to $\beta$, then $\beta$ is a descendant of $\alpha$. In particular, we have
\begin{equation}\label{LambdaGamma}
\Lambda_{[N]}^t(t-r) = \{\alpha : \Gamma_\alpha^{r} (t-r) \neq \varnothing \}.
\end{equation}
In general, $\Lambda_{[N]}^t(t-r) \neq \Gamma_{[N]}^{r}(t-r)$ realisationwise (although they are equal in distribution). However, it is true that 
\begin{equation}\label{LambdaGammavarnothing}
  \Lambda_{[N]}^{t}(t - r) = \varnothing \Longleftrightarrow \Gamma_{[N]}^{r}(t - r) = \varnothing
\end{equation}
for any $0 < r < t$, because if the population at time $t$ does not have ancestors  at time $r$, then the descendants of the population at time $r$ do not survive until $t$.
Moreover, if $t \gg r$ and the sets in \eqref{LambdaGammavarnothing} are not empty, we have
\begin{equation}\label{LambdaGammaNxi}
  \lvert\Lambda_{[N]}^{t}(t - r)\rvert \approx N \bar \xi \approx \lvert\Gamma_{[N]}^{r}(t - r) \rvert
\end{equation}
realisationwise due to metastability. 
In the realisation of Figure \ref{fig:numAnDe}, the equalities in \eqref{LambdaGammavarnothing} are true for $r = t_1$, $t = t_2$; that is, both processes go extinct in $[t_1,t_2]$. In contrast, for $r=t_2$ and $t = t_3$ or $t = t_4$,  \eqref{LambdaGammaNxi} applies.

Now, \eqref{LambdaGammavarnothing} and \eqref{LambdaGammaNxi}   together allow us to approximate $L_N^t$ of \eqref{Ln_in_path-integral} as
\begin{equation}\label{LNlin}
L_N^t \approx \int_{- \infty}^{t} \lvert \Gamma_{[N]}^{r}(t-r) \rvert \dd r 
= \int_0^{\infty} \lvert \Gamma_{[N]}^{t-\varrho}(\varrho) \rvert \dd \varrho
= \int_0^{S^t} \lvert \Gamma_{[N]}^{t-\varrho}(\varrho) \rvert \dd \varrho  \approx  N \bar \xi S^t,
\end{equation}
where $S^t$ denotes the largest $\varrho \leq t$ such that $\Gamma_{[N]}^{t-\varrho}(\varrho) \neq \varnothing$, or, equivalently, the smallest $\varrho \leq t$ such that $\Gamma_{[N]}^{(t-\varrho)-}(\varrho) = \varnothing$. 
This is the time\footnote{Note that $(\Gamma_{[n]}^t(\varrho))_{\varrho \geq 0}$ is based on the forward process, so it is c\`{a}dl\`{a}g; that is, a jump at time $t$ means that the `old' state applies until time $t-$ (the moment `just before' $t$), and at time $t$, the process is already in the `new' state. This implies that the ancestors are extinct at $(t-S^t)-$, but alive at $t-S^t$.}
since the last extinction event before $t$. 
In other words, the population at forward time $t$ has their oldest ancestor(s) at forward time $t-S^t$.
Since the extinction times are distinct random points on the time axis, $(S^t)_{t \geq 0}$ is a sawtooth function that increases linearly with slope 1 and is reset to some small value at random times.
So \eqref{LNlin} explains the linear increase in forward time, interrupted by collapses at the jumps of $(S^t)_{t \geq 0}$. 

Since $(\Gamma_{[N]}^t(\varrho))_{\varrho \geq 0}$ has the same law as $(\Lambda_{[N]}^t(\tau))_{\tau \geq 0}$, it is clear that $S^t$ has the same law as $T^t$, namely the exponential distribution with expectation $\EE[T]$ of \eqref{ET}. Eq.~\eqref{LNlin} therefore implies that $\EE[L_N] \approx N \bar \xi \EE[T]$. As a consistency check, we compared this numerically with the exact value of $\EE[L_N]$ of \eqref{ELn_culture_withd} for our parameter values $u=1,N=100$, and $s \in [0,1.4]$ and found good agreement (not shown).

We can now also better understand the dynamics of $(L_1^t)_{t \geq 0}$ and  $(C_1(t))_{t \geq 0}$  in the lower panel of Figures~\ref{fig:FullTimeSeries_super} and \ref{fig:SampleTimeSeries}. Since, for any given $\alpha$ and $t$, $(\Lambda^{t}_{\alpha}(\tau))_{\tau \geq 0}$ quickly either dies out  or gets close to $(\Lambda^{t}_{N}(\tau))_{\tau \geq 0}$, it is clear that, most of the time, $(L_1^t)_{t \geq 0}$ is either close to 0 or close to $(L_N^t)_{t \geq 0}$. The frequent, random transitions between the two possibilities come from the fact that, at the metastable state over time,  the actual set of lines that has old ancestors and thus ample knowledge, while being approximately constant in size, moves around in the population via the events of teaching individuals outside the current set, and by death events, thus rapidly including or excluding  individuals. That is, by learning from someone with old ancestors, an individual  acquires the entire ancestry and knowledge of the parent; and when the individual dies, it loses all  its ancestry, including all its traits.   

Let us note in passing that, for the parameter values in Figure \ref{fig:FullTimeSeries_super},
\eqref{ET} gives $\E{T} \approx 85.3$. Manual counting in Figure \ref{fig:FullTimeSeries_super}  yields $\approx 83$ sawteeth, which amounts to a mean tooth length (or extinction time) of 10000/83 $\approx$ 120. This overestimation of the length is presumably due to the underestimation of the number of teeth because small teeth are not resolved.

\subsection{Related concepts and phenomena in population genetics}
While the evolving genealogies considered here, and, in particular, their sawtooth dynamics, seem to be new, there are some  related concepts and phenomena in population genetics. We have already  mentioned the model of Muller's ratchet with tournament selection \citep{GCSmadiWakolbinger23}, where evolving  ASGs play a crucial role. In contrast to the ALG, these ASGs also have coalescence events, but the general behaviour is similar: on a longer time scale, there are repeated extinction events of the (so-called minimal-load) ASG (they correspond to the clicks of the ratchet). Between these events, the line-counting process of the ASG spends most of its time in a quasistationary state. On a shorter time scale between the extinction events, ASGs corresponding to different samples merge quickly, provided these samples are sufficiently large.

The connection to the  aforementioned evolving genealogies in Kingman's coalescent  \citep{PfaffelhuberWakolbinger2006,PfaffelhuberWakolbinger2011} is somewhat more remote. In contrast to the ALG, these genealogies are trees; their lengths also display a sawtooth behaviour (see \citet[Fig.~2]{PfaffelhuberWakolbinger2011}), but this stems from the change of roots (that is, the most recent common ancestor) of the trees (and subtrees), as opposed to the extinction of the ALG. Interestingly, \citet[Sec.~5]{PfaffelhuberWakolbinger2006} also superpose an innovation mechanism, namely infinite-sites mutation, to their genealogies, which is perfectly analogous to our invention of traits.  The mutations along the tree then translate into the differences between contemporary sequences, and the jumps of the most recent common ancestor  lead to  \emph{clusters} of such sequence differences, in line with a classical observation by \cite{Watterson1982}.


\subsection{Popularity spectrum}

\citet{aoki2018} (see also \cite{strimling2009tpb,
fogarty2015,fogarty2017}) studied the popularity spectrum $({P_l})_{l \in [N]}$ in discrete-time models of cultural evolution under various assumptions on the innovation and transmission mechanisms, where $P_l \defeq \E{C_{N,l}}$ and $C_{n,l}$ is the number of traits carried by exactly $l$ individuals in a sample of size $n$ at stationarity; this was obtained via the equilibrium condition in the forward model. In Appendix \ref{appendix:recursion}, we adapt this method to our model and derive 
\begin{equation}\label{PopSpectrum}
    {P_l} = \left( \frac{s}{Nu} \right)^{l-1}
    \frac{\mu \fall{N}{l}}{ul}, \quad \ell \in [N],
\end{equation}
in agreement with the simulations in Figure~\ref{fig:PopSpec}. Starting from  the $P_l$, we also extend  the trait frequency spectrum to a sample of size $n$ and obtain in Appendix \ref{appendix:recursion} that
\begin{equation}\label{PopSpectrumSample}
 \E{C_{n,i}} = \sum_{l=1}^{N} P_l  \frac{ \binom{l}{i} \binom{N-l}{n-i}}{ \binom{N}{n} }, \quad i \in [n].
\end{equation}
As a consistency check, or as an alternative derivation of $\E{C_{n}}$ via a forward approach, we also get $\E{C_{n}} = \sum_{i=1}^n \E{C_{n,i}} = \mu \E{L_n}$ of \eqref{moments_Cn} and \eqref{ELn_culture_withd} (note that $\E{C_{N}}$ is called $C_{\rm{pop}}$ in some of the aforementioned studies).

The bimodal distribution of the popularity spectrum in the supercritical case observed in Figure \ref{fig:PopSpec} reflects the by-now-familiar   fact that some traits die out quickly while the others survive for a long time, with very few intermediate traits, recall Figure \ref{fig:heatmapLandC}.

We would like to note that a bimodal popularity spectrum is not tied to joint transmission of traits ($b=1$). It also occurs in the model of \citet{NakamuraWakanoAokiKobayashi2020} with independent transmission ($b<1$) when an individual may learn both from its genetic parent (vertical transmission) and from any other individual (horizontal transmission). Bimodality is, therefore, more than just an artifact of $b=1$.

\begin{figure}
    \centering    \includegraphics[width=100mm]{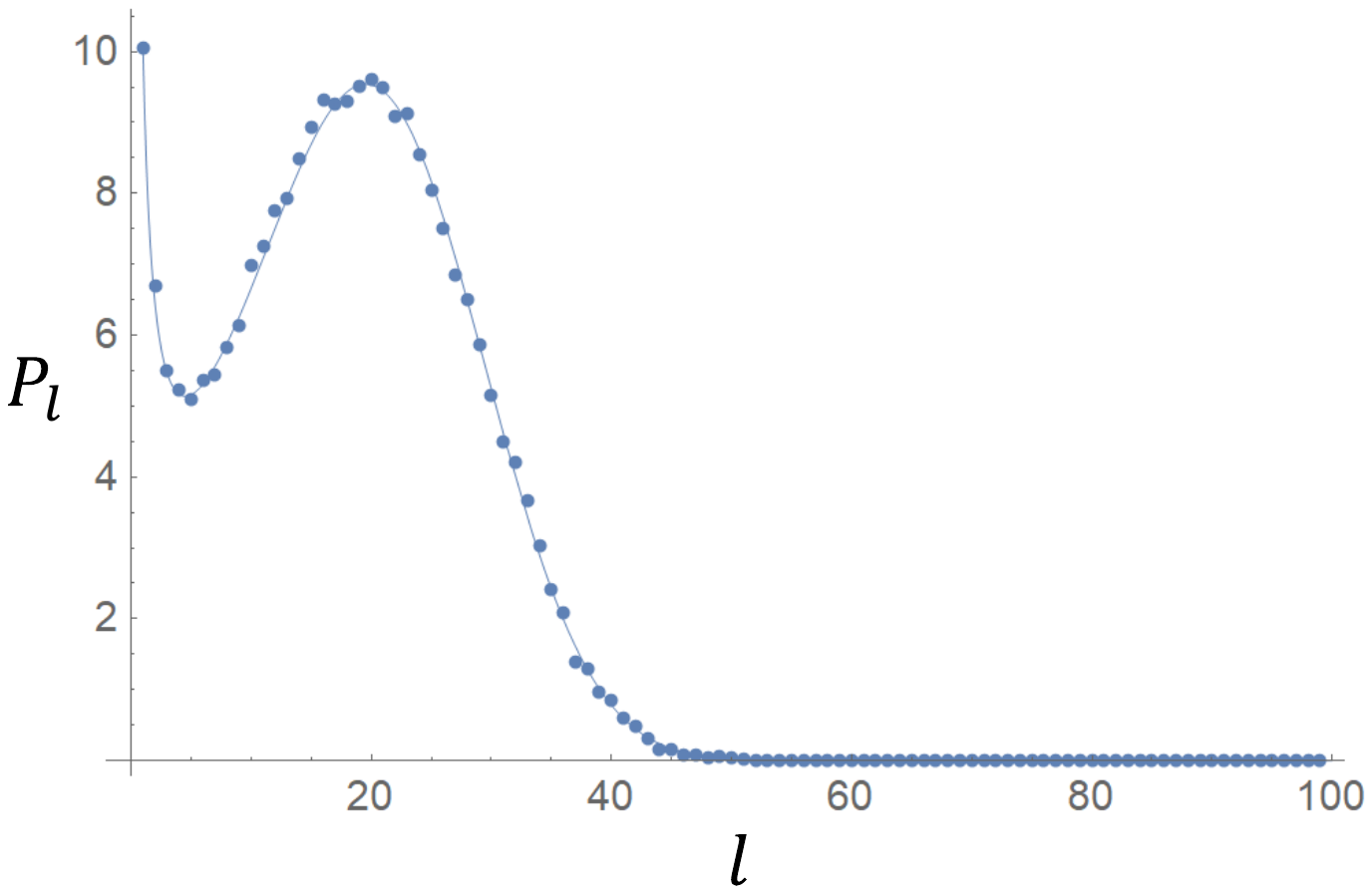}
    \caption{Popularity spectrum obatained from a long forward simulation run (bullets) compared with the analytic formula \eqref{PopSpectrum} (solid line). $s=1.3, N=100, \mu=0.1, u=1, t_{\rm{max}}=10^5$.}
    \label{fig:PopSpec}
\end{figure}


\section{Discussion}
Mathematical models for the accumulation of cultural traits in a population over time are crucial for understanding the cultural variation in  human populations. Here we have investigated cultural dynamics in which selectively neutral, discrete cultural traits are invented by individuals independently of one another at a constant rate and transmitted between individuals through random social learning. 
Related previous analyses \citep{strimling2009tpb,LehmannAokiFeldman2011,fogarty2015,fogarty2017,aoki2018,NakamuraWakanoAokiKobayashi2020} were mostly static and centered around the traits, focussing on the expectation of quantities like the number of distinct cultural traits maintained in a population at stationarity. Our approach is different in three ways: first, we concentrate on the underlying genealogies (as first described by \citet{aguilar2015}) and thus obtain results that remain true independently of the traits. Second, we investigate the dynamical aspects, which seem to be unexplored so far. At the heart of this  are the time evolution of the set of ancestors of a single individual or the entire population; the merging process between the various ancestral sets; the metastability of the corresponding counting processes;
and the concept of evolving genealogies. Third, we also obtain moments of $L_n$ and $C_n$.

The most conspicuous feature of our model is the sawtooth-like behaviour of $(L_N^t)_{t \geq 0}$ in the supercritical case, which is reflected in $(C_N(t))_{t \geq 0}$. The former is an inherent (and, to the best of our knowledge, novel) property of the evolving genealogy and deserves further mathematical investigation. In contrast, the latter is a consequence of the  transmission model  assumed here for simplicity:  in every learning 
event, all traits carried by the role model are 
transmitted to the learner, and the traits, once learnt, are never forgotten until death strikes the carrier. 
Thus, traits are never detached from one another once they come together
in an individual. This ``sticky'' nature of traits  implies that (nearly) ignorant individuals immediately become knowledgeable once they  learn from a knowledgeable parent. Despite the somewhat artificial assumption, the model may capture certain aspects of the accumulation of technologies or knowledge in human populations. 
For example, techniques for traditional craftwork are usually products of the accumulation of improvements 
over many generations and are transmitted as clusters through close apprenticeship;  but they can be irreversibly lost if all the carriers happen to die or no successors are found.

Indeed, it is increasingly recognised based on evidence as well as theory that cumulative cultural evolution of human technologies does not occur in a monotonic manner, but via phases of gradual accumulation 
of innovations  punctuated by sudden changes like rapid cascades of innovations (``leaps'') 
or drastic loss \citep{KolodonyCreanzaFeldman2015,VidiellaCarrignonBentleyETAL2022}.
In particular, the sudden loss of sophisticated or complex technologies and subsequent replacement by 
degraded ones in ethnographic or archaeological records is often  attributed
to demographic factors such as population bottlenecks or the fragmentation of  social networks, see, for example, \citet{Henrich2004} or
\citet{JacobsRobers2009}. Our results demonstrate that such behaviour is, in principle, also possible in the absence of demographic changes, via the inherent stochastic properties of the evolving genealogies.

As an outlook, let us nevertheless extend the model to allow for independent transmission of traits, similar to the discrete-time model of  \cite{KobayashiWakanoOhtsuki2018}, by assuming that each trait of the role model is independently transmitted to the learner with probability $b$ ($b=1$ reproduces our original model). 

To keep the mean number of traits transmitted via learning events constant, we replace $s$ by $s/b$.
For $b=0.99$ and, hence, our usual choice $s=1.3$ replaced by $s=130/99$, $(c_N(t))_{t \geq 0}$ still follows $\mu (\ell_N^t)_{t \geq 0}$ closely, see Figure \ref{fig:LNt_beta} (top).   For $b=13/14$ and, hence, $s=1.3$ replaced by $s=1.4$ (Figure \ref{fig:LNt_beta} bottom), 
the average size of the sawteeth in $(\ell_N^t)_{t \geq 0}$ is significantly higher, in line with the steep increase of $\E{L_N}$ with $s$, see Figure \ref{fig:changing_s}. 
More importantly, the  dynamics of $(c_N(t))_{t \geq 0}$ does not mirror that of $(\ell_N^t)_{t \geq 0}$ any more. More precisely,  a collapse of $(\ell_N^t)_{t \geq 0}$ still implies a  collapse of $(c_N(t))_{t \geq 0}$, simply because an extremely short genealogy provides no opportunity to accumulate innovations, no matter how they are transmitted. However, $(c_N(t))_{t \geq 0}$ remains below $(\mu \ell_N^t)_{t \geq 0}$, because inheritance of traits is increasingly diluted with their age.

\begin{figure} 
    \centering
    \includegraphics[width=160mm] 
    {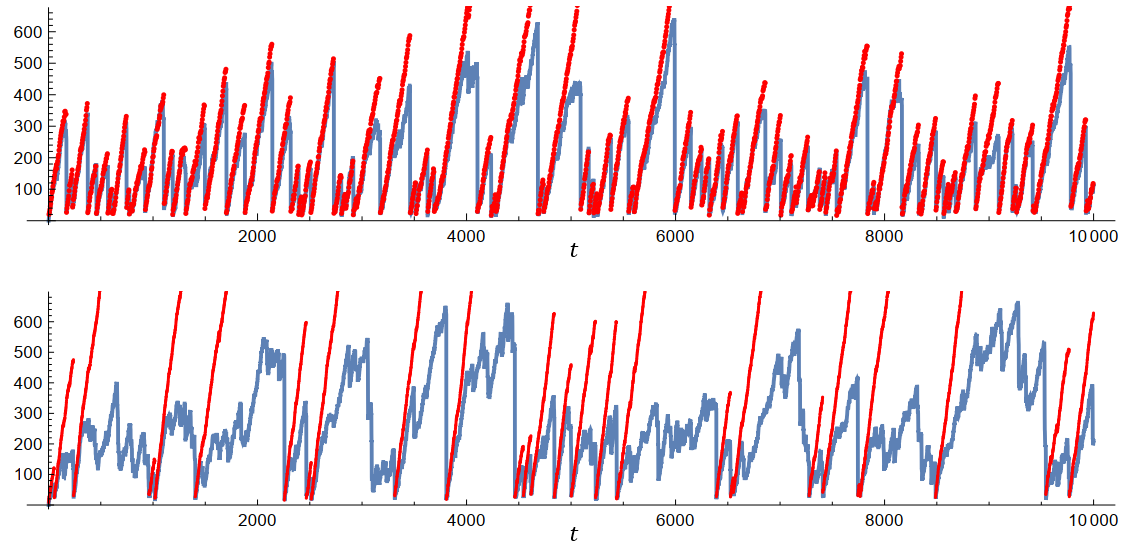}
    \caption{Simulated time series  $(\mu \ell_N^t)_{t \in [t_{\rm{max}}]_0}$ (red) and $(c_N(t))_{t \in [t_{\rm{max}}]_0}$ (blue) with $b<1$. Upper panel: $s=130/99,b=99/100$. Lower panel: $s=14/10,b=13/14$. Common parameters: $N=100, \mu=0.1, u=1, t_{\rm{max}} = 10000$. Where  no red line is drawn, it exceeds 700.
    }
    \label{fig:LNt_beta}
 \end{figure}

The  aim of the present study was to develop  the genealogical theory of cultural evolution, highlighting its unique features as well as close relationship with ancestral structures in population genetics. The cultural genealogies form the basis upon which arbitrary trait transmission models may be superposed; ours was chosen for its mathematical appeal rather than for realism or even data analysis.

For the more realistic transmission model mentioned above (with independent transmission both horizontal and vertical), \cite{NakamuraWakanoAokiKobayashi2020} have compared the predicted popularity spectrum  to the empirical equivalent extracted from datasets of cultural variation between archaeological sites or ethno-linguistic groups. Such datasets record the variation in, for example, bead types in Palaeolithic Europe, pottery in Neolithic Europe, and Fijian and Polynesian canoe design  in binary format (presence and absence of each character) at the level of groups of humans   rather than single ones. 
(The number of groups in these data sets ranges from 10 to 195. 
Incidentally, therefore,  $N=100$ as chosen for computational practicality in our simulations seems to be appropriate.)  
A note of caution is in order, however: if there is an inherent temporal dynamics, as in our model, then the pure equilibrium consideration and use of cross-sectional data alone may cause misleading results. However, the  temporal resolution required for a comparison with the \emph{dynamics} of learning models cannot be expected in archeological or ethnographic data.  Possibly, digital media (e.g., social network services) may  provide such data for future work.

\newpage

\section*{Acknowledgment}
We are greatly indebted to Anton Wakolbinger (Frankfurt) for enlightening discussions about  Muller's ratchet with tournament selection and many other conceptual and technical aspects of the paper, including a more elegant  proof of Proposition A.1. We also thank Tom Britton (Stockholm) for a helpful conversation about the connection to epidemiology. Two anonymous reviewers dedicated a lot of time and expertise to the manuscript and made invaluable suggestions to improve it. J.Y.W. received funding from MEXT/JSPS KAKENHI Grant Numbers JP21K03357 and JP24H00001. H.O.~acknowledges the support from the "Evolutionary Studies of Complex Adaptive Systems" Research Grant from the Research Center for Integrative Evolutionary Science, SOKENDAI. Y. K. received funding from MEXT/JSPS KAKENHI Grant Numbers 25K06712 and 25K00199. E.B. received funding from the Deutsche Forschungsgemeinschaft,
Germany (DFG, German Research Foundation) --- Project-ID 317210226
--- SFB 1283.

\newpage

\setcounter{equation}{0}
\setcounter{section}{0}

\numberwithin{equation}{section} 

\renewcommand{\thesection}{\Alph{section}}

\renewcommand{\theequation}{\thesection\arabic{equation}} 

\section*{Appendix}

\section{Sojourn times of the birth-death process}\label{sec:sojourn}
Consider our general continuous-time birth-death process $(Z(t))_{t \geqslant 0}$ on $[N]_0$  with unique absorbing state  0 and birth and death rates $\lambda_{j}$ and $\mu_{j}$ with $\lambda_{0}=0$, $\lambda_j \geq 0$ for $j \in [N-1]$, and $\mu_j>0$ for $j \in [N]$, complemented  by   $\mu_{0}= \lambda_N=0$.  

\begin{prop}
Let $Z(0)=i \in [N]$, let $T_\ell$ be the time where the process hits state $ \ell \; (1 \leq \ell \leq i)$ for the first time, and let $S_j^{(\ell)}$ be the total sojourn time in state $j \; (\ell \leq j \leq N)$ in the interval $[T_\ell, T_{\ell-1}]$. 
We then have
\begin{equation*}
    \EE \big [ S_j^{(\ell)}  \mid Z(0)=i \big ] = 
        \frac{\lambda_{\ell} \cdots \lambda_{j-1}}{\mu_{\ell} \cdots \mu_{j}} \eqdef \eta^{}_{\ell j} 
\end{equation*}
independently of $i$, where the empty product is 1.
\end{prop}

\begin{proof}
The $T_\ell$ are finite almost surely due to the almost sure absorption of $(Z(t))_{t \geqslant 0}$ in 0. If $\lambda_{\ell} \cdots \lambda_{j-1}=0$, the claim is trivially true. We therefore assume $\lambda_{\ell} \cdots \lambda_{j-1}>0$. Let $V_{j}$ be the number of visits  the chain pays to state $j$ in $[T_\ell, T_{\ell-1}]$ (this includes the initial visit if $j=\ell$). Each such visit  has a mean duration of $1/(\lambda_j+\mu_j)$. By Wald's identity, we have
\begin{equation}\label{Wald}
    \EE \big [ S_j^{(\ell)} \mid Z(0)=i \big ] =  \frac{\E{V_j} }{\lambda_j + \mu_j}.
\end{equation}
To compute $\E{V_j}$, we define a \emph{discrete-time} birth-death process $(Z'(n))_{n \geq 0}$ on  $\{\ell-1, \ell, \dots, N\}$ with $Z'(0)=\ell-1$ and transition probabilities 
\begin{equation}\label{transprob}
         p_{\ell-1, \ell}  = 1, \quad \text{together with } \, p_{k, k-1} = \frac{\mu_{k}}{\lambda_{k}+\mu_{k}} \quad 
        \text{and } \, p_{k,k+1} = \frac{\lambda_{k}}{\lambda_{k}+\mu_{k}} \quad \text{for } \, \ell \leq k \leq N.
\end{equation}
Then $T'_\ell=1$ is the first time where $(Z'(n))_{n \geq 0}$ reaches $\ell$, and we denote by $T'_{\ell-1}$ the time of its first visit to $\ell-1$. Clearly, $(Z'(n))_{n \geq 0}$ restricted to $[T'_{\ell},T'_{\ell-1}]$ has the same law as the embedded jump process of $(Z(t))_{t \geq 0}$ restricted to $[T_{\ell},T_{\ell-1}]$, so  $V'_{j}$, the number of visits of $(Z'(n))_{n \geq 0}$ to $j$ in  $[T'_{\ell},T'_{\ell-1}]$, has the same law as $V_j$, and, in particular, $\E{V_j} = \E{V'_j}$. We now complement the $V_j'$ ($\ell \leq j \leq N$) by $V'_{\ell-1} \defeq 1$ (for the initial visit of $(Z'(n))_{n \geq 0}$ to $\ell-1$) and note that, for $\ell \leq j \leq N$, the final step of $(Z'(n))_{n \geq 0}$ to $\ell-1$ does not count into $V_j'$. This gives $V'_k=\sum_{n=0}^{T_{\ell-1}-1} \bs{1}_{\{Z(n)=k\}}$ for $\ell-1 \leq k \leq N$, and $v \defeq (\E{V'_{\ell-1}}, \ldots, \E{V'_N})$ constitutes an invariant measure of $(Z'(n))_{n \geq 0}$, see \citet[Thm.~1.7.5]{Norris1997}. On the other hand, $(Z'(n))_{n \geq 0}$ is recurrent and has a unique stationary distribution, $u=(u_{\ell-1}, u_{\ell}, \ldots, u_{N})$ say; it satisfies the detailed balance condition $u_{k} p_{k,k+1} = u_{k+1} p_{k+1,k}$ for $\ell-1 \leq k <N$. By Thm.~1.7.6 of \citet{Norris1997}, $v$ is proportional to $u$, so
\begin{equation}
    \E{V'_j} = \frac{\E{V'_j}}{\E{V'_{\ell-1}}} = \frac{u_{j}}{u_{\ell-1}} = \frac{p_{\ell-1, \ell} \cdots p_{j-1, j}}{p_{\ell, \ell -1} \cdots p_{j, j-1}} = \frac{\lambda_{\ell} \cdots \lambda_{j-1}}{\mu_{\ell} \cdots \mu_{j}} (\lambda_{j} + \mu_{j}),
\end{equation}
where the first equality comes from $V'_{\ell-1}=1$, the second-last from  detailed balance, and the last one from \eqref{transprob}. This, combined with $\EE(V_j)=\EE(V_j')$ and \eqref{Wald}, completes the proof.
\end{proof}

We would like to emphasise that \citet{Stefanov1995} has proved the analogous result for finite-state discrete- or continuous-time birth-death processes, not necessarily absorbing, with the help of analytical properties of certain exponential families. We have complemented this here by a simple probabilistic argument.

\section{A derivation of the moments of the path integrals of birth-death processes}\label{sec:moment}

The moments \eqref{mth_moment_formula} of the path integral \eqref{def_path-integral} are usually derived via Laplace transforms (see, for example, \citet[App.~D]{goel1974stochastic}).  
Here we take an alternative route that bypasses Laplace transforms and derives a first-step equation in a direct way, thus validating (62) derived  in a heuristic way by \cite{norden1982distribution} via the backward equation.
Remember that our birth-death process $(Z(t))_{t \geqslant 0}$ is restricted to $[N]_0$ with birth and death rates $\lambda_{j} \geq 0$ and $\mu_{j}>0$ for $j \in [N]$, $\lambda_0=0$, and the convention $\mu_{0}^{}= \lambda_{N}^{}=0$.

Let us now calculate the $m$-th  moment of our path integral $\mathcal{I}_{i}(f)$ of \eqref{def_path-integral}, following the standard first-step approach outlined, for example, in \citet[Sec.~3.4]{KarlinPinsky2011}. Suppose that we are currently in state $i \in [N]$. The sojourn time in this state, whose realisation is denoted by $r$ below, follows the exponential distribution with parameter $\lambda_{i}+\mu_{i}$. This time period contributes $r f(i)$ to the path integral. After the transition away from $i$, the further contribution to the path integral until absorption  is distributed as $\mathcal{I}_{i+1}(f)$ or $\mathcal{I}_{i-1}(f)$, depending on whether it moves from $i$ to  $i+1$ (probability $\lambda_{i}/(\lambda_{i}+\mu_{i})$) or  $i-1$ (probability $\mu_{i}/(\lambda_{i}+\mu_{i})$).  Thus we obtain
\begin{equation} \label{exp_Tim_new}
	\begin{split}
		& \E{(\mathcal{I}_{i}(f))^{m}} \\
        &=  \int_{0}^{\infty} \left\{ \frac{\lambda_{i}}{\lambda_{i}+\mu_{i}} \E{\{\mathcal{I}_{i+1}(f) + r f(i)\}^{m}} + \frac{\mu_{i}}{\lambda_{i}+\mu_{i}} \E{\{\mathcal{I}_{i-1}(f) + r f(i)\}^{m}} \right\} (\lambda_{i}+\mu_{i}) e^{-(\lambda_{i}+\mu_{i})r} \mathrm{d}r \\
        &= \lambda_{i} \E{ \int_{0}^{\infty} \{\mathcal{I}_{i+1}(f) + r f(i)\}^{m} e^{-(\lambda_{i}+\mu_{i})r} \mathrm{d}r} + \mu_{i} \E{ \int_{0}^{\infty} \{\mathcal{I}_{i-1}(f) + r f(i)\}^{m} e^{-(\lambda_{i}+\mu_{i})r} \mathrm{d}r}
	\end{split}
\end{equation}
(products containing the factor $\lambda_N$ are $0$).
The integrals in the second line  can be evaluated by integration by parts as
\begin{equation}
    \begin{split}
        & \int_{0}^{\infty} \{\mathcal{I}_{i \pm 1}(f) + r f(i)\}^{m} e^{-(\lambda_{i}+\mu_{i})r} \mathrm{d}r \\
        &= \left[ -\{\mathcal{I}_{i \pm 1}(f) + r f(i)\}^{m} \frac{e^{-(\lambda_{i}+\mu_{i})r}}{\lambda_{i}+\mu_{i}}\right]_{0}^{\infty} + m f(i) \int_{0}^{\infty} \{\mathcal{I}_{i \pm 1}(f) + r f(i)\}^{m-1} \frac{e^{-(\lambda_{i}+\mu_{i})r}}{\lambda_{i}+\mu_{i}} \mathrm{d}r \\
        &= \frac{(\mathcal{I}_{i \pm 1}(f))^{m}}{\lambda_{i}+\mu_{i}} + \frac{m f(i)}{\lambda_{i}+\mu_{i}} \int_{0}^{\infty} \{\mathcal{I}_{i \pm 1}(f) + r f(i)\}^{m-1} e^{-(\lambda_{i}+\mu_{i})r} \mathrm{d}r,
    \end{split}
\end{equation}
so  $\E{(\mathcal{I}_{i}(f))^{m}}$ turns into
\begin{equation}
    \begin{split}
         \E{(\mathcal{I}_{i}(f))^{m}}  
        = \, & \frac{\lambda_{i}}{\lambda_{i}+\mu_{i}} \E{ (\mathcal{I}_{i+1}(f))^{m} } + \frac{\mu_{i}}{\lambda_{i}+\mu_{i}} \E{ (\mathcal{I}_{i-1}(f))^{m} } \\
        & + \frac{m f(i)}{\lambda_{i}+\mu_{i}} \Bigg\{ \lambda_{i} \E{ \int_{0}^{\infty} \{\mathcal{I}_{i+1}(f) + r f(i)\}^{m-1} e^{-(\lambda_{i}+\mu_{i})r} \mathrm{d}r} \\
        &\phantom{\frac{m f(i)}{\lambda_{i}+\mu_{i}} \Bigg\{ } + \mu_{i} \E{ \int_{0}^{\infty} \{\mathcal{I}_{i-1}(f) + r f(i)\}^{m-1} e^{-(\lambda_{i}+\mu_{i})r} \mathrm{d}r} \Bigg\}.
    \end{split}
\end{equation}
Since the expression in curly brackets equals  $\E{(\mathcal{I}_{i}(f))^{m-1}}$ by \eqref{exp_Tim_new}, we obtain the first-step equation
\begin{equation}\label{moment_equation_new}
	(\lambda_{i} + \mu_{i}) \E{ (\mathcal{I}_{i}(f))^{m} } =  \lambda_{i} \E{ (\mathcal{I}_{i+1}(f))^{m} } + \mu_{i} \E{ (\mathcal{I}_{i-1}(f))^{m} } + m f(i) \E{ (\mathcal{I}_{i}(f))^{m-1} }, \; m\geq 1, \,  i \in [N],
\end{equation}
with boundary conditions $\E{ (\mathcal{I}_{0}(f))^{m} } = 0$ for all $m \geq 1$ and $\E{ (\mathcal{I}_{i}(f))^{0} } = 1$ for all $i \in [N]$. 

Since \eqref{moment_equation_new} expresses  $m$-th  moments in terms of $(m-1)$-st  moments, we can now generalise \cite{norden1982distribution} and recursively solve \eqref{moment_equation_new} from lower to higher  moments.
To this end, let  $\bs{b}^{(m)} \defeq \left( \E{ (\mathcal{I}_{1}(f))^{m} }, \cdots, \E{ (\mathcal{I}_{N}(f))^{m} } \right)^{\top}$ for $m\geq 0$; note that $\bs{b}^{(0)} = \bs{1}$, a vector whose  components are all one. Then \eqref{moment_equation_new} reads 
\begin{equation}\label{moment_equation_matrix_form_1_new}
    \bs{A} \bs{b}^{(m)} = m \bs{F} \bs{b}^{(m-1)}, \quad m>0,
\end{equation}
where $\bs{A}=(A_{ij})_{i,j \in [N]}$ is an $N \times N$ matrix with elements
\begin{equation}
{A}_{ij} =
    \begin{cases}
        \lambda_{i} + \mu_{i}, & j=i, \\
        -\lambda_{i}, & j=i+1, \\
        -\mu_{i}, & j=i-1, \\
        0, & \text{otherwise},
    \end{cases}
\end{equation}
and $\bs{F}$ is the $N \times N$ diagonal matrix  $\bs{F}= \textrm{diag}\left[ f(1), \cdots, f(N)\right]$. Therefore, from \eqref{moment_equation_matrix_form_1_new} we immediately obtain
\begin{equation}\label{moment_equation_matrix_form_2_new}
     \bs{b}^{(m)} = m \left(\bs{A}^{-1} \bs{F}\right) \bs{b}^{(m-1)},
\end{equation}
which leads to
\begin{equation}
\bs{b}^{(m)} = m  \left(\bs{A}^{-1} \bs{F}\right) \bs{b}^{(m-1)} = \cdots = m!  \left(\bs{A}^{-1} \bs{F}\right)^{m} \bs{b}^{(0)} = m!  \left(\bs{A}^{-1} \bs{F}\right)^{m} \bs{1}.
\end{equation}
 From the standard theory of Markov chains (\citealt[Chap. III]{KemenySnell1960}; see also \citealt[p. 694]{norden1982distribution}), we know that $\bs{A}$ is invertible and $\bs{A}^{-1}$ has elements  $\left( \bs{A}^{-1} \right)_{ij} = \zeta_{ij}$  with $\zeta_{ij}$ of \eqref{t_ij_formula} and \eqref{s_ellj_formula};  so \eqref{mth_moment_formula} is an immediate consequence.

\section{Derivation of the first two moments of  $L_{n}$}\label{sec]branch_length}
\subsection{First moment of $L_{n}$}
Inserting \eqref{t_ij_formula_SIS} into \eqref{1st_moment_formula} with $f=\mathrm{id}$ leads to 
\begin{equation}\label{ELn_culture_pre}
     \E{L_{n}} =\sum_{j=1}^{N} \sum_{\ell=1}^{\min(n,j)} \frac{1}{u} \left(\frac{s}{Nu}\right)^{j-\ell} \fall{(N-\ell)}{j-\ell}.
\end{equation}
Substituting $m=j-\ell$ and changing  summation turns this into
\begin{equation}\label{ELn_culture_withd_pre}
    \E{L_{n}} = \sum_{m=0}^{N-1} \sum_{\ell=1}^{\min(n, N-m)} \frac{1}{u} \left(\frac{s}{Nu}\right)^{m} \fall{(N-\ell)}{m}.
\end{equation}
The latter expression is further simplified by using 
\begin{equation}\label{discrete_analysis_formula}
    \sum_{i=a}^{b} \fall{i}{m} = \frac{\fall{(b+1)}{m+1} - \fall{a}{m+1}}{m+1},
\end{equation}
which holds for $a,b,m \in \NN$ with $a \leq b$ 
(see Chapter 2 of \citealt{Graham1994concrete}).
In our case, set $a=N-M$ with $M=\min(n, N-m)$ and $b=N-1$. If $n \leq N-m$, then $\fall{(N-M)}{m+1}=\fall{(N-n)}{m+1}$. If, on the other hand,  $n > N-m$,  we have $\fall{(N-M)}{m+1}=\fall{m}{m+1}=0$, but also  $\fall{(N-n)}{m+1}=0$ because $m>N-n$. So, in any case, $\fall{(N-M)}{m+1}=\fall{(N-n)}{m+1}$, and
\begin{equation}\label{discrete_analysis_formula_modified}
    \sum_{\ell=1}^{\min(n, N-m)} \fall{(N-\ell)}{m} = \sum_{i=N-M}^{N-1} \fall{i}{m} = \frac{\fall{N}{m+1} - \fall{(N-M)}{m+1}}{m+1} = \frac{\fall{N}{m+1} - \fall{(N-n)}{m+1}}{m+1},
\end{equation}
which allows us to carry out the second sum in \eqref{ELn_culture_withd_pre} and yields
\begin{equation}\label{ELn_culture_withd2}
    \E{L_{n}} = \sum_{m=0}^{N-1} \frac{1}{u} \left(\frac{s}{Nu}\right)^{m} \frac{\fall{N}{m+1} - \fall{(N-n)}{m+1}}{m+1}.
\end{equation}

\subsection{Second moment of $L_{n}$}
A similar calculation enables us to derive the second moment. Using  \eqref{mth_moment_formula} and \eqref{t_ij_formula_SIS} with $f=\textrm{id}$, we have
\begin{equation}\label{ELn2_culture_withd_pre}
    \begin{split}
        \E{L_{n}^{2}} &= 2 \sum_{j_{1}=1}^{N} \sum_{\ell_{1}=1}^{\min(n,j_{1})} \frac{1}{u} \left(\frac{s}{Nu}\right)^{j_{1}-\ell_{1}} \fall{(N-\ell_{1})}{j_{1}-\ell_{1}} \sum_{j_{2}=1}^{N}  \sum_{\ell_{2}=1}^{\min(j_{1},j_{2})} \frac{1}{u} \left(\frac{s}{Nu}\right)^{j_{2}-\ell_{2}} \fall{(N-\ell_{2})}{j_{2}-\ell_{2}}
        \\
        &= 2 \sum_{m_{1}=0}^{N-1} \sum_{\ell_{1}=1}^{\min(n, N-m_{1})} \frac{1}{u} \left(\frac{s}{Nu}\right)^{m_{1}} \fall{(N-\ell_{1})}{m_{1}} \sum_{m_{2}=0}^{N-1} \sum_{\ell_{2}=1}^{\min(m_{1}+\ell_{1}, N-m_{2})}\frac{1}{u} \left(\frac{s}{Nu}\right)^{m_{2}} \fall{(N-\ell_{2})}{m_{2}},
    \end{split}
\end{equation}
where we have set $m_{1}=j_{1}-\ell_{1}$ and $m_{2}=j_{2}-\ell_{2}$ in the second step.
By  applying \eqref{discrete_analysis_formula_modified} in the first and the last step, it is further evaluated to 
\begin{equation}\label{ELn2_culture_withd_derivation}
    \begin{split}
        \E{L_{n}^{2}}
        &= 2 \sum_{m_{1}=0}^{N-1} \sum_{\ell_{1}=1}^{\min(n, N-m_{1})} \frac{1}{u} \left(\frac{s}{Nu}\right)^{m_{1}} \fall{(N-\ell_{1})}{m_{1}} \\
        & \qquad \times\sum_{m_{2}=0}^{N-1} \frac{1}{u} \left(\frac{s}{Nu}\right)^{m_{2}} \frac{\fall{N}{m_{2}+1} - \fall{\{N-(m_{1}+\ell_{1})\}}{m_{2}+1}}{m_{2}+1} \\
        &= 2 \sum_{m_{1}=0}^{N-1} \sum_{m_{2}=0}^{N-1} \left(\frac{1}{u}\right)^{2} \left(\frac{s}{Nu}\right)^{m_{1}+m_{2}} \\ 
        & \qquad \times \sum_{\ell_{1}=1}^{\min(n, N-m_{1})}   \frac{\fall{(N-\ell_{1})}{m_{1}} \fall{N}{m_{2}+1} - \fall{(N-\ell_{1})}{m_{1}} \fall{\{N-(m_{1}+\ell_{1})\}}{m_{2}+1}}{m_{2}+1} \\
        &= 2 \sum_{m_{1}=0}^{N-1} \sum_{m_{2}=0}^{N-1} \left(\frac{1}{u}\right)^{2} \left(\frac{s}{Nu}\right)^{m_{1}+m_{2}} \sum_{\ell_{1}=1}^{\min(n, N-m_{1})}   \frac{\fall{(N-\ell_{1})}{m_{1}} \fall{N}{m_{2}+1} - \fall{(N-\ell_{1})}{m_{1}+m_{2}+1}}{m_{2}+1} \\
        &= 2 \sum_{m_{1}=0}^{N-1} \sum_{m_{2}=0}^{N-1} \left(\frac{1}{u}\right)^{2} \left(\frac{s}{Nu}\right)^{m_{1}+m_{2}} \\ 
        & \qquad \times \left[ \frac{\{\fall{N}{m_{1}+1} - \fall{(N-n)}{m_{1}+1}\} \fall{N}{m_{2}+1}}{(m_{1}+1)(m_{2}+1)} - \frac{\fall{N}{m_{1}+m_{2}+2} - \fall{(N-n)}{m_{1}+m_{2}+2}}{(m_{1}+m_{2}+2)(m_{2}+1)} \right]. \\
    \end{split}
\end{equation}


\section{Derivation of the popularity spectrum}\label{appendix:recursion}
Recall that ${P_l}$ is the expected number of traits carried by exactly $l$ individuals in an equilibrium population of size $N$. We calculate $(P_l)_{l \in [N]}$ by adapting the method of \citet{aoki2018} (see also \cite{strimling2009tpb,
fogarty2015,fogarty2017}) to continuous time as follows.
 Death occurs at rate $u$ to every individual, and when there are $l$ individuals carrying the trait (that is, the trait has popularity $l$), this trait loses a carrier at death rate $d_l=ul$, upon which its popularity decreases by one. Every individual that carries a trait with popularity $l$ transmits this trait to some individual that lacks the trait at rate $s (N-l)/N$, so the trait gains a carrier and its popularity increases by one; altogether, therefore, we have birth events at rate $b_l=s l(N-l)/N$ when in state $l\in [N]$. Innovation events produce new traits of popularity 1 at rate $N\mu$ without affecting any other trait. Therefore, the $P_l$ obey the  system of differential equations 
\begin{equation}\begin{split}
     \dot P_1 & = N\mu +  d_2 P_2 -(b_1+d_1)P_1, \\
    \dot P_l &= b_{l-1} P_{l-1} +  d_{l+1} P_{l+1} - (b_l+d_l)P_l,
     \quad 2 \le l \le N,
\end{split}\end{equation}
with the convention $d_{N+1}=P_{N+1}=0$; note that the structure is the same as that of the Kolmogorov forward equation for a birth-death process with immigration to state 1.  The (unique) stationary solution is given by $b_l P_l=d_{l+1}P_{l+1}$, $l \in [N-1]$, together with $N\mu=uP_1$, which jointly give
\begin{equation}
    {P_l} = \frac{b_1 b_2 \cdots b_{l-1}}{d_2 d_3 \cdots d_l} {P_1} \hspace{30 pt}
    \text{for } \; 2 \le l \le N, \quad \text{complemented by} \quad P_1 =  \frac{N\mu}{u},
\end{equation}
or, more explicitly,
\begin{equation}\label{P_l}
    {P_l} = \left( \frac{s}{Nu} \right)^{l-1}
    \frac{\prod_{i=1}^{l-1}  i(N-i)}{l!} {P_1}
     = \left( \frac{s}{Nu} \right)^{l-1}
    \frac{\mu \fall{N}{l}}{ul},
\end{equation}
which is \eqref{PopSpectrum}. 

We can extend this method to obtain the trait frequency spectrum in a sample of size $n<N$. 
For  a trait with popularity $l$, the number of  individuals in a sample of size $n$ that carry the trait follows a hypergeometric distribution; more precisely, the probability that $i$ out of the $n$ individuals carry the trait is  $\binom{l}{i} \binom{N-l}{n-i}/\binom{N}{n}$.
If $C_{n,i}$ is the number of traits carried by exactly $i$ individuals in a sample of size $n$, then the trait frequency spectrum in a sample of size $n$ results as
\begin{equation}\label{ECni}
 \E{C_{n,i}} = \sum_{l=1}^{N} P_l  \frac{ \binom{l}{i} \binom{N-l}{n-i}}{ \binom{N}{n} },
\end{equation}
which is \eqref{PopSpectrumSample}.
As a consistency check, note that
\begin{equation}\label{ECn}
 \E{C_n} =  \sum_{i=1}^{n}  \E{C_{n,i}}
\end{equation}
by definition, and, since
\[
  \sum_{i=1}^n \frac{ \binom{l}{i} \binom{N-l}{n-i}}{ \binom{N}{n} } = 
  1- \frac{ \binom{l}{0} \binom{N-l}{n}}{ \binom{N}{n} } = 1- \frac{ \fall{(N-l)}{n} }{ \fall{N}{n} }
   = 1- \frac{ \fall{(N-n)}{l} }{ \fall{N}{l} } 
\]
(where the first step comes from the normalisation of the hypergeometric distribution and the last step is true because $\fall{(N-l)}{n} \fall{N}{l}  = \fall{N}{n+l}
=  \fall{(N-n)}{l} \fall{N}{n}$), \eqref{ECni} and \eqref{ECn} together give 
\begin{equation}
    \E{C_{n}}   
    =\sum_{l=1}^{N} P_l
    \left(  1- \frac{ \fall{(N-n)}{l} }{ \fall{N}{l} } \right) = \mu \E{L_n},
\end{equation}
where the last step uses \eqref{P_l} and \eqref{ELn_culture_withd}; the result agrees with $\E{C_n}$ of \eqref{moments_Cn}. Specifically, for
$n=N$, we get $\E{C_{N}} = \sum_{l=1}^N P_l$ for the expected total number of distinct traits in the population (denoted by $C_{\rm{pop}}$ in some of the previous studies), as it must.

\bibliographystyle{chicago}

\end{document}